\documentclass[a4paper,UKenglish,cleveref, autoref, thm-restate]{lipics-v2021}
\pdfoutput=1 %uncomment to ensure pdflatex processing (mandatatory e.g. to submit to arXiv)
\hideLIPIcs  %uncomment to remove references to LIPIcs series (logo, DOI, ...), e.g. when preparing a pre-final version to be uploaded to arXiv or another public repository

\usepackage{amsmath}
\usepackage{breqn}
\usepackage{amssymb}
\usepackage{amsthm}
\usepackage{mathtools}
\usepackage{wrapfig}
\usepackage{bbm} % for the characteristic function \mathbbm{1}

\usepackage[linesnumbered,lined,boxed,commentsnumbered]{algorithm2e}
\usepackage{adjustbox}

\usepackage{tikz}
\usetikzlibrary{fit}

\usetikzlibrary{external}
\tikzstyle{empty}=[fill=white, draw=black, shape=circle]

\tikzstyle{original}=[-, fill=black, draw=black]
\tikzstyle{oriented}=[fill=black, draw=red, ->]

\newcommand{\E}[1]{\mathbb{E} \left[ #1 \right]}
\newcommand{\var}[1]{\text{Var} \left( #1 \right)}
\newcommand{\cov}[2]{\text{Cov} \left( #1, #2 \right)}
\newcommand{\prob}[1]{\mathbb{P}\left[#1\right]}
\newcommand{\bigo}[1]{\mathcal{O}\left(#1\right)}

\SetKwProg{Fn}{Function}{}{}
\SetKwFunction{TriangleCounting}{TriangleCounting}
\SetKwFunction{OddCycleCounting}{OddCycleCounting}
\SetKwFunction{GetOrdering}{GetOrdering}
\SetKwFunction{GraphProjection}{GraphProjection}
\SetKwFunction{RR}{RR}
\SetKwFunction{Lap}{Lap}

\title{Cycle Counting under Local Differential Privacy for Degeneracy-bounded Graphs} % Please add

\author{Quentin Hillebrand}{The University of Tokyo, Tokyo, Japan}{quentin-hillebrand@g.ecc.u-tokyo.ac.jp}{https://orcid.org/0000-0002-7747-4998}{}
\author{Vorapong Suppakitpaisarn}{The University of Tokyo, Tokyo, Japan}{vorapong@is.s.u-tokyo.ac.jp}{https://orcid.org/0000-0002-7020-395X}{}
\author{Tetsuo Shibuya}{The University of Tokyo, Tokyo, Japan}{tshibuya@hgc.jp}{https://orcid.org/0000-0003-1514-5766}{}
\authorrunning{Q. Hillebrand, V. Suppakitpaisarn, and T. Shibuya}

\ccsdesc[500]{Security and privacy~Privacy-preserving protocols}
\ccsdesc[300]{Theory of computation~Graph algorithms analysis}
\keywords{Differential privacy, triangle counting, degeneracy, arboricity, graph theory, parameterized accuracy} % mandatory; please add comma-separated list of keywords

\relatedversion{} %optional, e.g. full version hosted on arXiv, HAL, or other respository/website
\nolinenumbers %uncomment to disable line numbering

\EventEditors{John Q. Open and Joan R. Access}
\EventNoEds{2}
\EventLongTitle{42nd Conference on Very Important Topics (CVIT 2016)}
\EventShortTitle{CVIT 2016}
\EventAcronym{CVIT}
\EventYear{2016}
\EventDate{December 24--27, 2016}
\EventLocation{Little Whinging, United Kingdom}
\EventLogo{}
\SeriesVolume{42}
\ArticleNo{23}
\begin{document}

\maketitle

\begin{abstract}
We propose an algorithm for counting the number of cycles under local differential privacy for degeneracy-bounded input graphs. Numerous studies have focused on counting the number of triangles under the privacy notion, demonstrating that the expected \(\ell_2\)-error of these algorithms is \(\Omega(n^{1.5})\), where \(n\) is the number of nodes in the graph. When parameterized by the number of cycles of length four (\(C_4\)), the best existing triangle counting algorithm has an error of \(O(n^{1.5} + \sqrt{C_4}) = O(n^2)\). 
In this paper, we introduce an algorithm with an expected \(\ell_2\)-error of \(O(\delta^{1.5} n^{0.5} + \delta^{0.5} d_{\max}^{0.5} n^{0.5})\), where \(\delta\) is the degeneracy and \(d_{\max}\) is the maximum degree of the graph. For degeneracy-bounded graphs (\(\delta \in \Theta(1)\)) commonly found in practical social networks, our algorithm achieves an expected \(\ell_2\)-error of \(O(d_{\max}^{0.5} n^{0.5}) = O(n)\).
Our algorithm's core idea is a precise count of triangles following a preprocessing step that approximately sorts the degree of all nodes. This approach can be extended to approximate the number of cycles of length \(k\), maintaining a similar \(\ell_2\)-error, namely $O(\delta^{(k-2)/2} d_{\max}^{0.5} n^{(k-2)/2} + \delta^{k/2} n^{(k-2)/2})$ or $O(d_{\max}^{0.5} n^{(k-2)/2}) = O(n^{(k-1)/2})$ for degeneracy-bounded graphs.
\end{abstract}

\newpage

\section{Introduction}

%The recent multiplication of interconnected system has motivated the development of graph-based analysis as the best representation for this type of data. Although this provides a powerful framework and valuable insights, the sensitivity of the represented information calls for caution when it comes to the publication of statistics.

In recent years, differential privacy~\cite{dwork2006differential,dwork2014algorithmic} has become the gold standard for providing strong privacy guarantees while enabling meaningful data analysis. Differential privacy ensures that the output of a computation does not significantly change when any single individual's data is modified, thus safeguarding individual privacy. While much of the initial work in differential privacy focused on traditional tabular data~\cite{dwork2006calibrating,mcsherry2007mechanism}, there is increasing interest in extending these privacy guarantees to graph data~\cite{sajadmanesh2021locally,ye2020lf}, which presents its own unique set of challenges.

Differential privacy has evolved into numerous variants to accommodate different scenarios, as detailed in \cite{desfontaines2019sok}. Of particular interest to us is the concept of local differential privacy~\cite{cormode2018privacy,evfimievski2003limiting}. This variant is unique in that it does not rely on the assumption of a trusted central server. Instead, users must obfuscate their private data before sharing it with an untrusted computing entity. In the context of graph data, the most commonly adopted notion is edge local differential privacy~\cite{qin2017generating}, where the sensitive information of each user pertains to their connections with others.

A widely used obfuscation method is randomized response \cite{warner_randomized_1965, wang2016using}. In this approach, users invert each bit of their adjacency vector with a certain probability. The server then collects this distorted information to construct an obfuscated graph. Although it is possible to publish various graph statistics from the obfuscated graph, the resulting information tends to be imprecise. Algorithms specifically designed to publish particular statistics typically yield more accurate and useful graph information.

\begin{table}[h!]
\centering
\begin{adjustbox}{width=\textwidth}
\begin{tabular}{|c|c|c|}
\hline
 & \textbf{Upper Bound} & \textbf{Lower Bound} \\
\hline
\textbf{Triangle} & $O(n^2)$ (\cite{imola2021locally}, general graphs) & $\Omega(n^{1.5})$ (non-interactive) \cite{imola2022differentially} \\
 & $O(d_{\mathrm{max}}^{1.5}n^{0.5})$ (\cite{eden2023triangle}, general graphs) & $\Omega(n^{1.5})$ (interactive) \cite{eden2023triangle} \\
 & $O(d_{\mathrm{max}}^{0.5} n^{0.5})$ (\textcolor{red}{this work}, degeneracy-bounded graph) & $\Omega(n^{2})$ (non-interactive) \cite{eden2023triangle}\\
\hline
\textbf{Odd length} & $O\left(n^{k - 1}\right)$ (folklore, general graphs) & \\ \textbf{cycles $C_k$}
 & $O\left(n^{(k - 1)/2}\right)$ (\textcolor{red}{this work}, degeneracy-bounded graph) & \\
\hline
\end{tabular}
\end{adjustbox}
\caption{Upper and lower bounds of the expected $\ell_2$-error for triangle and $k$-cycle counting under the local differential privacy}
\label{tab1}
\end{table}

One graph statistic frequently considered by researchers in local differential privacy is the number of subgraphs \cite{imola2021locally,hillebrand2023communication}. Specifically, many studies have focused on the publication of triangle counts \cite{imola2021locally,imola2022communication,hillebrand2023communication,eden2023triangle}. Theoretical analysis results on the $\ell_2$-error are summarized in Table~\ref{tab1}. Unfortunately, to date, when $n$ is the number of nodes and $d_{\max}$ is the maximum degree of the input graphs, the best algorithm has an expected $\ell_2$-error of $O(n^2)$ or $O(d_{\max}^{1.5}n^{0.5})$. We believe that this error is too large for many applications and should be improved. On the other hand, it has been shown that for all locally differentially private algorithms, there exists a class of graphs where the $\ell_2$-error is $\Omega(n^{1.5})$ \cite{eden2023triangle}. This lower bound implies that the expected $\ell_2$-error cannot be significantly improved.

\subsection{Our Contribution}

This motivates us to consider a specific class of graphs. Specifically, in this paper, we focus on graphs with bounded degeneracy, a characteristic common to most practical social networks \cite{wasa2018efficient,goel2006bounded,eppstein2013listing}. Let the graph degeneracy be $\delta$. We propose a locally differentially private algorithm with an expected $\ell_2$-error of $O\left(\delta^{1.5}n^{0.5} + \delta^{0.5} d_{\max}^{0.5} n^{0.5} \right)$. When the graph degeneracy is bounded ($\delta = O(1)$), the expected $\ell_2$-error becomes $O(d_{\max}^{0.5} n^{0.5}) = O(n)$. This result implies that our expected error for the degeneracy-bounded graphs can be smaller than the lower bound for general graphs. 

We also extend our results to count the number of cycles with odd lengths in degeneracy-bounded graphs. To our knowledge, there are only two local differentially private algorithms proposed for counting subgraphs of more than three nodes. The first algorithm \cite{imola2022differentially} is designed to count the number of four-length cycles but operates within the shuffle model, which is weaker than the original local differential privacy model. The second algorithm counts the number of walks of length $k$ \cite{betzerpublishing}. This field has limited work due to the significant noise introduced to ensure user privacy, which accumulates as the subgraph size increases. This accumulation results in unacceptable errors for differential privacy in larger subgraphs. For instance, while the expected $\ell_2$-error from triangle counting algorithms based on randomized response is $O(n^2)$ \cite{imola2022differentially}, the expected $\ell_2$-error for similar algorithms estimating the number of $C_k$ is as high as $O(n^{k - 1})$. In other words, the error increases by a factor of $n$ with each increment in cycle length. 

In this work, we propose an algorithm that significantly reduces the expected \(\ell_2\)-error to \(O(n^{(k - 1)/2})\) in degeneracy-bounded graphs. We believe that this error is much smaller than the actual number of cycles in most graphs. Consequently, our algorithm is the first to publish a meaningful number of large cycles under local differential privacy.

\subsection{Technical Overview}

In this section, we provide an overview of the technical concepts behind our triangle counting algorithm. The algorithm for counting odd-length cycles, for \( k \geq 5 \), extends these ideas but requires a more intricate and detailed analysis.

Let the input graph be \( G = (V = \{\nu_1, \dots, \nu_n\}, E) \). In prior work \cite{imola2021locally}, they apply a randomized response mechanism that flips each bit in the adjacency matrix with a certain probability. Let the resulting graph after applying the randomized response be \( G' = (V, E') \). In the local differential privacy setting, each node \( \nu_i \) knows whether it is connected to another node \( \nu_j \) (where \( \nu_j \neq \nu_i \)) if \( \{\nu_i, \nu_j\} \in E \). For the triangle counting method, node \( \nu_i \) considers \( (\nu_i, \nu_j, \nu_\kappa) \) as a triangle if \( \{\nu_i, \nu_j\} \in E \), \( \{\nu_i, \nu_\kappa\} \in E \), and \( \{\nu_j, \nu_\kappa\} \in E' \). 
Define \( e_{i,j,\kappa} = 1 \) if node \( \nu_i \) considers \( (\nu_i, \nu_j, \nu_\kappa) \) as a triangle, otherwise set \( e_{i,j,\kappa} = 0 \). Define $S_i = \{(j,\kappa): \{\nu_i, \nu_j\}, \{\nu_i, \nu_\kappa\} \in E \text{ and } j < \kappa\}$.
The estimated number of triangles for node \( \nu_i \), reported by the user, is \( \tilde{t}_i = \sum\limits_{(j,\kappa) \in S_i} e_{i,j,\kappa} \). The total estimated number of triangles in the graph is then \( \tilde{f}_\Delta(G) = \frac{1}{3} \sum_i \tilde{t}_i = \frac{1}{3} \sum\limits_{i} \sum\limits_{(j,\kappa) \in S_i} e_{i,j,\kappa} \), where the division by three accounts for each triangle being counted three times by each of the users involving in the triangle. 

The \(\ell_2\)-error of the estimated triangle count \(\tilde{f}_\Delta(G)\) mostly arises from the variance in the estimation. A significant portion of this variance comes from the covariance between pairs of variables in the summation \( \frac{1}{3} \sum\limits_{i} \sum\limits_{(j,\kappa) \in S_i} e_{i,j,\kappa} \). Two variables, \(e_{i,j,\kappa}\) and \(e_{i',j',\kappa'}\), are dependent if \((j,\kappa) = (j',\kappa')\). 
The number of dependent pairs in the counting process is equivalent to the number of tuples \((\nu_i, \nu_j, \nu_{i'}, \nu_\kappa)\) such that \((j,\kappa) \in S_i \cap S_{i'}\), which corresponds to the number of 4-cycles in the input graph \(G\). Therefore, the squared \(\ell_2\)-error is approximately proportional to the number of 4-cycles in the graph, which is $O(n^4)$.

Let us assume that the indices of all users are predetermined and publicly known before the counting process begins. Define \( S'_i = \{(j, \kappa) : \{\nu_i, \nu_j\}, \{\nu_i, \nu_\kappa\} \in E \text{ and } j < i < \kappa \} \). If node \( i \) only considers the pairs \( (j, \kappa) \) within \( S'_i \), then each triangle is counted exactly once. The estimated number of triangles, \( \hat{f}_\Delta(G) \), can be calculated as \( \hat{f}_\Delta(G) = \sum_i \hat{t}_i \), where \( \hat{t}_i = \sum\limits_{(j,\kappa) \in S'_i} e_{i,j,\kappa} \). In this counting method, the number of dependent variable pairs is at most the number of 4-cycles that contain the three nodes \( \nu_i, \nu_j, \nu_\kappa \) with \( j < i < \kappa \).

Let \(\delta\) represent the degeneracy of the input graph \(G\), and for each \(\nu \in V\), let \(d(\nu)\) denote the degree of \(\nu\). Assume that the degrees of all nodes are publicly known, and the nodes are indexed in non-decreasing order of their degree, i.e., if \(i > j\), then \(d(\nu_i) \leq d(\nu_j)\). Referring to the bound established by Chiba and Nishizeki \cite{chiba1985arboricity}, which states that \(\sum\limits_{(\nu_i, \nu_j) \in E} \min(d_i, d_j) \leq O(\delta \cdot |E|)\), we demonstrate in this paper that the number of such cycles is \(O(\delta^3 n)\). Consequently, the squared \(\ell_2\)-error is reduced from \(O(n^4)\) in previous work to \(O(\delta^3 n)\).

However, we cannot assume that the degrees of all nodes are publicly known, as this information is sensitive. To address this issue, we use local Laplacian queries, allowing each user to publish a noisy version of their degree. Let the noisy degree of \(\nu \in V\) be denoted as \(\tilde{d}(\nu)\). We then assign indices to users based on these noisy degrees, such that if \(i > j\), then \(\tilde{d}(\nu_i) \leq \tilde{d}(\nu_j)\). Afterward, we run the protocol described in the previous paragraph. We show that even with noisy degrees, the expected number of such cycles remains bounded by \(O(\delta^3 n)\).

In summary, our mechanism involves two steps. First, users publish their noisy degrees using the local Laplacian mechanism, and the server assigns indexes based on these noisy values. In the second step, using the results of randomized response, each user \(\nu_i\) estimates the number of triangles \((\nu_i, \nu_j, \nu_\kappa)\) where \(j < i < \kappa\). This method significantly reduces the number of dependent triangle pairs in degeneracy-bounded graphs, which in turn lowers the variance of the estimation.

\subsection{Related Works}

The field of graph data mining under local differential privacy is relatively new. In contrast, differential privacy has been studied for many years by various researchers, including works like \cite{gupta2010differentially,olatunji2021releasing}. According to \cite{imola2021locally}, local differential privacy typically only hides edges or relationships, except in special cases like \cite{zhang2020differentially}. Differential privacy, on the other hand, can hide whether an individual or node is part of a social network, as shown in \cite{hay2009accurate,raskhodnikova2016differentially}. Therefore, while both edge and node differential privacy exist, node differential privacy does not apply in the context of local differential privacy.

%In addition to counting cycles, previous works have proposed methods to estimate the number of \(k\)-stars \cite{imola2021locally,hillebrand2023communication} and the variance of degrees \cite{hillebrand2023communication}. These graph statistics can be determined by applying certain functions to node degrees, requiring only degree information from each user. Therefore, we believe these tasks are less challenging than counting the number of cycles. There is also work on counting triangles that explores scenarios where users have access to their 2-hop adjacency information \cite{liu2022collecting}. In this paper, we focus on 1-hop adjacency information, which is more commonly considered. 

Recent works have proposed methods to estimate the densest subgraph, \(k\)-core decomposition, and degeneracy under local differential privacy \cite{dhulipala2022differential, dhulipala2023near, dinitz2023improved}. However, since we are focused on estimating different graph statistics in graphs, we do not use or extend the ideas from these works. Instead, the estimation of degeneracy can be used to approximate the \(\ell_2\)-error of our algorithm.

%introduced an approximation method for identifying the densest subgraph using locally adjustable algorithms. \cite{} improved upon this by eliminating the multiplicative error in the \(k\)-core decomposition. Additionally, \cite{dinitz2023improved} describes a method to find the densest subgraph up to an additive error under the edge-local differential privacy framework. However, these works primarily focus on estimating these statistics. We believe that these statistics could be used to parameterize the complexity of private graph algorithms, similar to their use in non-private algorithms \cite{chiba1985arboricity}.

%The strong motivation for the focus on degeneracy-bounded graphs comes from how accurate they are to represent real-world graphs. In fact, most graphs have been observed to have low degeneracy \cite{eppstein2013listing}. Our intuition is that this property can be leveraged to provide more accurate private algorithms for subgraph counting, in particular for triangle counting. To make the case for this claim, we will describe a new triangle counting algorithm and prove that its complexity can be expressed in terms of the degeneracy of the graph.

\section{Preliminaries}
\label{sec:preliminaries}

\subsection{Notations}

For \( V = \{\nu_1, \ldots, \nu_n\} \) a set of vertices and \( E \subseteq V^2 \) a set of edges, we denote by \( G = (V,E) \) the graph on \( V \). We consider simple undirected graphs, meaning that for \( \nu, \nu' \in V \), \( (\nu, \nu) \not\in E \) and \( (\nu, \nu') \in E \implies (\nu', \nu) \in E \). We denote by \( n = |V| \) the size of the graph and \( m = |E| \) its number of edges.

For each \( i \in [1,n] \), we introduce \( a_i = [a_{i,1}, \ldots, a_{i,n}] \), the adjacency list of user \( \nu_i \), where for any \( j \in [1,n] \), \( a_{i,j} = 1 \) if the edge \( (\nu_i, \nu_j) \) is in \( E \) and \( a_{i,j} = 0 \) otherwise. Additionally, we introduce \( d_i \), the degree of node \( \nu_i \), which corresponds to the number of edges incident to \( \nu_i \).

We call a path of length \( k \in \mathbb{N} \), denoted \( P_k \), any tuple \( (\nu_{l_1}, \ldots, \nu_{l_k}) \) such that, for all \( i \in [1,k] \), \( (\nu_{l_i}, \nu_{l_{i+1}}) \in E \), and, for all $i \neq j$, $\nu_{l_i} \neq \nu_{l_j}$. We also use \( \#P_k(G) \) to refer to the number of paths of length \( k \) in \( G \). Similarly, a cycle of length \( k \in \mathbb{N} \), or \( C_k \), is a tuple \( (\nu_{l_1}, \ldots, \nu_{l_k}) \) that forms a path and satisfies \( (\nu_{l_k}, \nu_{l_1}) \in E \). We will also use \( \#C_k(G) \) to refer to the number of cycles of length \( k \) in \( G \). %Furthermore, we set \( P_{-1} = 1 \) by convention to simplify the equations.

\subsection{Edge Local Differential Privacy}

We say that two adjacency list $a$ and $a'$ are neighboring if they differ by one bit, i.e. if we can go from one to the other by adding or removing an edge to node $\nu_i$. If $a'$ is a neighbor of $a$, we write that $a \sim a'$. %When this is the case, we will write $d(a_i, a'_i) = 1$.
The notion of edge local differential privacy is as follows:
\begin{definition}[$\varepsilon$-edge local differentially private query] 
\label{ref1} Let $\varepsilon > 0$. A randomized algorithm $\mathcal{R}$ is a $\varepsilon$-edge local differentially private query on the node $\nu_i$ if, for all neighboring bit strings $a \sim a'$, and for all $S$, it holds that 
\begin{equation*}
    \prob{\mathcal{R}(a) \in S} \leq e^{\varepsilon} \prob{\mathcal{R}(a') \in S}.
\end{equation*} 
\end{definition}
\begin{definition}[$\varepsilon$-edge local differentially private algorithm \cite{qin2017generating}] \label{dif:diff}
Let \(\mathcal{A}\) be an algorithm that generates multiple randomized queries for each user, has each user apply these queries to their adjacency vector, and then estimates some graph statistics based on the results. We say \(\mathcal{A}\) is an \(\varepsilon\)-edge local differentially private algorithm if, for all users \(\nu_i\) and for all possible sets of queries \(\mathcal{R}_1, \dots, \mathcal{R}_k\) inquired to $\nu_i$ (where for each \(1 \leq j \leq k\), \(\mathcal{R}_j\) is an \(\varepsilon_j\)-edge local differentially private query), it holds that \(\varepsilon_1 + \cdots + \varepsilon_k \leq \varepsilon\). 
\end{definition}

%\subsection{Basic Mechanisms}

%We will now introduce the 2 private building blocks that we will base our algorithms on.
%The first one, the Laplacian Mechanism can be used to privately publish the result of a function by adding so noise tailored to the sensitivity of the function. The version we introduce here is the simpler case where the range of the function is in $\R$

%Numerous mechanisms have been suggested to meet Definition~\ref{dif:diff} \cite{li2022network,hou2023wdt}.
%Among them, the edge local Laplacian mechanism \cite{hillebrand2023unbiased} is an approach to offering privacy in which we ask each user to give real numbers to the aggregator. The Laplacian mechanism can be defined as in the following definition. It can be shown that the mechanism is $\varepsilon$-edge local differential private.

\subsection{Laplacian Query and Restricted Sensitivity}

Next, we introduce queries that are $\varepsilon$-edge local differentially private. We first consider a query which aims to give an estimate of a real number statistics of the adjacency vector.
%The edge local Laplacian query aims to give an estimate of a real number statistics of the adjacency vector, while the randomized response query aims to give an obfuscated adjacency list.

\begin{definition}[Edge local Laplacian query \cite{hillebrand2023unbiased}]
    For a function \( f: \{0,1\}^n \rightarrow \mathbb{R} \) on adjacency lists, and \( a \sim a' \) denoting neighboring adjacency lists, the global sensitivity of \( f \) is defined as
$\Delta_f = \max\limits_{a \sim a'} |f(a) - f(a')|.$
For \(\varepsilon > 0\), the query that outputs \( f(a) + \Lap(\Delta_f / \varepsilon) \) is \(\varepsilon\)-edge local differentially private, where \(\Lap(b)\) represents noise drawn from the Laplacian distribution with parameter \( b \). \label{laplacian}
\end{definition}

Global sensitivity in Definition \ref{laplacian} is designed to handle the worst-case scenario, which can lead to large amounts of noise being added to the data when using the Laplacian mechanism. However, if the data is known to belong to a specific set, restricted sensitivity allows us to adjust the noise according to the sensitivity within that set, resulting in more tailored and potentially lower noise levels.

\begin{definition}[Restricted sensitivity (Definition 8 of \cite{blocki2013differentially})]
    Let $a = (a_1, \dots, a_n), a' = (a'_1, \dots, a'_n) \in \{0,1\}^n$ and $d(a, a')$ be the Hamming distance between $a$ and $a'$. The restricted sensitivity of $f$ over a set of possible output $\mathcal{H}$ is
            \begin{equation*}
                RS_f(\mathcal{H}) = 
        \max\limits_{a,a' \in \mathcal{H}} \left(
            \frac{
                |f(a) - f(a')|
            }{
                d(a, a')
            }
        \right).
            \end{equation*}
\end{definition}

We can use restricted sensitivity to publish data even if it is not initially in the set. To do this, we first need to define a projection method to map the data to the set. In this work, we will consider \(\mathcal{H}_d\), the class of adjacency list with a maximum degree of \(d\), for calculating restricted sensitivity. We assume that the order of all nodes is fixed, and if a node \(\nu_i\) is adjacent to more than \(d\) nodes, we retain only the first \(d\) nodes according to this order. The map can be considered as an operation on each adjacency vector $a_i$. We denote the mapping result on $a_i$ as $\mu_d(a_i)$.
\begin{definition}[Edge local Laplacian query with restricted sensitivity on $\mathcal{H}_d$ \cite{blocki2013differentially}]
    \label{def:restricted}
    For any $f$ queried to a user $i$, the query that answers $f(\mu(a_i)) + \Lap(3 \cdot RS_f(\mathcal{H}_d) / \varepsilon)$ is called edge local Laplacian query with restricted sensitivity on $\mathcal{H}_d$, and provides $\varepsilon$-edge local differential privacy.
\end{definition}

%The second mechanism we will use is randomized response. It permits to privately release an entire list of bits, in our case the whole adjacency list.

%When we request users to send a vector of bits (such as the whole adjacency vector) to the aggregator, the edge local Laplacian mechanism is not the most suitable. Instead, a frequently used method is the randomized response \cite{warner_randomized_1965, mangat_improved_1994,wang2016using}, as detailed in the following definition.

\subsection{Unbiased Randomized Response}
In this subsection, we consider the randomized response query, which aims to publish an obfuscated adjacency vector.

\begin{definition}[Randomized response query \cite{warner_randomized_1965,wang2016using}]
For \(\varepsilon > 0\), the randomized response mechanism takes an adjacency list \(a = (a_1, \dots, a_n)\) as input and outputs an obfuscated list \(\tilde{a} = (\tilde{a}_1, \dots, \tilde{a}_n) \). For  \(i\), the probability that \(\tilde{a}_i\) is set to 1 is given by:
\[
\prob{\RR(\tilde{a}_i) = 1} = 
\begin{cases}
    \frac{e^{\varepsilon}}{1 + e^{\varepsilon}} & \text{if } a_i = 1 \\
    \frac{1}{1 + e^{\varepsilon}} & \text{if } a_i = 0.
\end{cases}
\]
With this definition, randomized response provides \(\varepsilon\)-edge local differential privacy.
\end{definition}

We can construct a graph \(\tilde{G}\) based on the collection of obfuscated adjacency vectors obtained from all users. Using the statistics of the obfuscated graph \(\tilde{G}\), we can then publish various information, including the number of subgraphs \cite{ye2020towards,imola2021locally,imola2022communication,hillebrand2023communication}. However, randomized response produces biased results, making it less suitable for counting queries. Those bias can be fixed by the subsequent definition.

\begin{definition}[Unbiased randomized response query \cite{eden2023triangle}]
    \label{def:unbiased}
Let \(\varepsilon > 0\) and \(\tilde{a}_i\) be the adjacency vector published through randomized response with budget \(\varepsilon\) by user \(\nu_i\). Then, for all \((i,j) \in [1,n]^2\), \begin{equation*}
    \hat{a}_{i,j} = \frac{e^{\varepsilon} + 1}{e^{\varepsilon} - 1} \tilde{a}_{i,j} - \frac{1}{e^{\varepsilon} - 1}
\end{equation*} is an unbiased estimator of \(a_{i,j}\). Additionally, for \((i,j) \neq (i',j')\), \(\hat{a}_{i,j}\) is independent of \(\hat{a}_{i',j'}\), and \(\var{\hat{a}_{i,j}} = \frac{e^{\varepsilon}}{(e^{\varepsilon} - 1)^2}\). We refer to a query that publishes \(\hat{a}_i\) as the unbiased randomized response query. This query is \(\varepsilon\)-edge locally differentially private.
\end{definition}

We can use the results from the unbiased randomized response query to calculate the number of subgraphs. For example, without privacy constraints, the number of triangles can be calculated as \(\sum\limits_{i < j < k} a_{i,j} \cdot a_{j,k} \cdot a_{k,i}\). To privately estimate the number of triangles, we use \(\sum\limits_{i < j < k} \hat{a}_{i,j} \cdot \hat{a}_{j,k} \cdot \hat{a}_{k,i}\). It is theoretically shown in \cite{eden2023triangle} that the estimator \(\sum\limits_{i < j < k} \hat{a}_{i,j} \cdot \hat{a}_{j,k} \cdot \hat{a}_{k,i}\) has a smaller \(\ell_2\)-error compared to the estimator obtained from the randomized response query, \(\sum\limits_{i < j < k} \tilde{a}_{i,j} \cdot \tilde{a}_{j,k} \cdot \tilde{a}_{k,i}\). 

\subsection{Graph Arboricity and Degeneracy}

%Several characterization of the sparsity of a graph can be found in the literature such as its number of edges, its maximal degree, its densest component, its degeneracy and its arboricity.
%For this work, we will focus on the last two notions and give our results in terms of accuracy parameterized by them.

%Our reason for choosing those two related notions is that arboricity has long been used outside of the differential private community for subgraph counting (see \cite{chiba1985arboricity} for example), while degeneracy has recently received some interest in the community \cite{dhulipala2022differential,dhulipala2023near}.

Graph arboricity and degeneracy can be defined as follows:

\begin{definition}[Arboricity]
    The arboricity of a graph $G$ is the minimal number $\alpha(G)$ such that the edges of $G$ can be partitioned into $\alpha(G)$ forests.
\end{definition}

\begin{definition}[Degeneracy]
    The degeneracy of a graph $G$ is the smallest number $\delta(G)$ such that any subgraph of $G$, contains at least one node with induced degree at most $\delta(G)$.
\end{definition}

We observe that the variable $\delta$ is frequently used as a privacy parameter in differential privacy. However, since we do not consider that parameter in this paper, we choose to use $\delta$ to represent degeneracy, which is also a common convention. When the context is clear, we will drop the $G$ of the notation and simply write $\alpha$ and $\delta$. 
The two quantities are linked by the following theorem.

\begin{theorem}[equation 3 and lemma 2.2 of \cite{zhou2000graph}]
    \label{thm:arboricity-degeneracy}
    In any graph $G$, degeneracy and arboricity satisfy
    \(\alpha \leq \delta \leq 2 \alpha - 1.\)
\end{theorem}

The arboricity has previously been used outside of the differential private community to bound some graph statistics. A folklore useful result is that the number of edges in a graph is smaller than $\delta n$. Another well known result is as follows:

\begin{theorem}[Chiba-Nishizeki Bound \cite{chiba1985arboricity}]
    \label{thm:chiba-nishizeki}
    With $m = |E|$ and $d_i$ the degree of node $\nu_i$, then
    \begin{equation*}
     \sum\limits_{(\nu_i,\nu_j) \in E} \min{(d_i, d_j)} \leq m \alpha.   
    \end{equation*}
\end{theorem}

\section{Node-Reordered Graphs and Their Properties}
\label{sec:ordering}

%For the rest of the work we will consider a graph $G = (V,E))$ with $V$ its set of vertices and $E$ its set of edges. In addition, we will denote $n = |V|$ the size and $m = |E|$ the number of edges.

The first step of our mechanism is to order the vertices based on their estimated degree. The algorithm for this step is shown in Algorithm \ref{algo:vertex-ordering}. At line 2 of the algorithm, we privately publish the estimated degree. 
Under edge local differential privacy, the global sensitivity of the degree is 1. Therefore, we can use the Laplacian query (Definition \ref{laplacian}) with noise scaled to \(1 / \varepsilon_0\) to publish the degree, where \(\varepsilon_0\) is the privacy budget allocated to this step. We denote the estimated degree as
$\tilde{d}_i = d_i + \text{Lap}\left(1 / \varepsilon_0\right)$.

\begin{algorithm}
    \caption{Calculate a low degree ordering of a graph with respect to the estimated degree}
    \label{algo:vertex-ordering}
    
    \Fn{\GetOrdering}{
        \KwIn{Graph $G = (V,E)$, privacy budget $\varepsilon_0$}
        \KwOut{A low degree ordering $\phi$ of $G$ with respect to the estimated degree}
        \textbf{[User $i$]} Calculate and send $\tilde{d}_i \leftarrow d_i + \Lap(\frac{1}{\varepsilon_0})$ to the central server\\ 
        \textbf{[Server]} Let $\phi(i) = j$ if $\tilde{d}_i$ is the $j$-the largest number in $\tilde{d}_1, \dots, \tilde{d}_n$. Calculate $\phi(i)$ for all $i$\\
        \Return $\phi$\;
    }
\end{algorithm}

After publishing the estimated degrees, in Line 3, we assign an order \(\phi\) to the nodes based on their degrees, which we refer to as a \textit{low degree ordering}. For \(G = (\{\nu_1, \dots, \nu_n\}, E)\), we denote the reordered graph as \(G^{\phi} = (V^{\phi}, E^{\phi})\), where \(V^{\phi} = \{\eta_i \mid i \in [1,n]\}\) and \(\nu_i = \eta_{\phi(i)}\) for all \(i\). The edge set \(E^{\phi}\) is defined as \(\{(\eta_{\phi(i)}, \eta_{\phi(j)}) \mid (\nu_i, \nu_j) \in E \}\). We note that \(G\) and \(G^{\phi}\) are isomorphic, and thus have the same number of subgraphs.
We denote by $d_i(G^{\phi})$ the degree of $\eta_i$ in $G^{\phi}$ and $d_i^-(G^{\phi})$ the number of neighbors of node $\eta_i$ in the set $\{\eta_1, \dots, \eta_{i - 1}\}$.

\newcommand{\addmargin}{
  \node[fit=(current bounding box),inner ysep=2mm,inner xsep=1mm]{};
}

\begin{table}[tbp]
    \centering
    \begin{tabular}{|c|c|c|c|c|}
        \hline
        Symbol & $S_2^*$ & $P_k$ & $C_k$ & $C_{2k}^*$ \\ \hline
        Representation & 
        \begin{tikzpicture}[baseline=0.5cm]
            \node [style=empty] (0) at (0, 0) {};
            \node [style=empty] (1) at (-0.75, 1) {};
            \node [style=empty] (2) at (0.75, 1) {};
            \draw [style=oriented] (0) to (1);
            \draw [style=original] (0) to (2);
            \addmargin
        \end{tikzpicture}
        & 
        \begin{tikzpicture}[baseline=1cm]
            \node [style=empty] (0) at (0, 0) {};
            \node [style=empty] (1) at (1, 0.5) {};
            \node [style=empty] (2) at (0, 1) {};
            \node [style=empty] (3) at (1, 1.5) {};
            \node [style=empty] (4) at (0, 2) {};
            \draw [style=original] (4) to (3);
            \draw [style=original] (3) to (2);
            \draw [style=original] (2) to (1);
            \draw [style=original] (1) to (0);
            \addmargin
        \end{tikzpicture}
        &
        \begin{tikzpicture}[baseline=0cm]
            \node [style=empty] (0) at (0, 1) {};
            \node [style=empty] (1) at (-1, 0.25) {};
            \node [style=empty] (2) at (1, 0.25) {};
            \node [style=empty] (4) at (-0.6, -1) {};
            \node [style=empty] (5) at (0.6, -1) {};
            \draw [style=original] (2) to (0);
            \draw [style=original] (0) to (1);
            \draw [style=original] (1) to (4);
            \draw [style=original] (4) to (5);
            \draw [style=original] (5) to (2);
            \addmargin
        \end{tikzpicture}
        &
        \begin{tikzpicture}[baseline=0.12cm, scale=0.8]
            \node [style=empty] (0) at (-1, 0.75) {};
            \node [style=empty] (1) at (-1, -0.5) {};
            \node [style=empty] (2) at (0, -1.25) {};
            \node [style=empty] (3) at (1, -0.5) {};
            \node [style=empty] (4) at (1, 0.75) {};
            \node [style=empty] (5) at (0, 1.5) {};
            \draw [style=original] (5) to (0);
            \draw [style=original] (0) to (1);
            \draw [style=original] (5) to (4);
            \draw [style=original] (4) to (3);
            \draw [style=oriented] (1) to (2);
            \draw [style=oriented] (2) to (3);
            \addmargin
        \end{tikzpicture}
        \\ \hline
        Bound &
        \(\displaystyle \bigo{\delta^2 n}\) & 
        \(\displaystyle \bigo{\delta^{\lceil \frac{k}{2} \rceil} n^{\lfloor \frac{k}{2} \rfloor + 1}}\) &
        \(\displaystyle \bigo{\delta^{\lceil \frac{k}{2} \rceil} n^{\lfloor \frac{k}{2} \rfloor}}\) &
        \(\displaystyle \bigo{\delta^{k + 1} n^{k - 1}}\)
        \\ \hline
    \end{tabular}
    \caption{List of subgraphs analyzed in Section \ref{sec:ordering}, including their representations and bounds on their counts in the graph produced by Algorithm \ref{algo:vertex-ordering}. Oriented edges indicate directionality, with an arrow from $\nu_j$ to $\nu_i$ signifying that $j > i$.}
    \label{tab:subgraphs}
\end{table}

In the remainder of this section, we analyze the properties of graphs produced by the reordering. Specifically, our focus is on bounding the frequency of certain substructures within the reordered graph. A summary of the results from this section is provided in Table~\ref{tab:subgraphs}.

\begin{definition}[low star]
For \(k \in \mathbb{N}^*\), a low-\(k\)-star is a subgraph consisting of a central node and \(k\) neighboring nodes, where at least one of the neighboring nodes has an index smaller than that of the central node. We denote by \(S_k^*(G)\) the number of such subgraphs contained in a graph \(G\). \label{def:lowstar}
\end{definition}

\begin{theorem}
    \label{thm:s2-bound}
        $\mathbb{E}\left[S_2^*(G^{\phi})\right] \leq \bigo{\delta^2 n}$.
 \end{theorem}

\begin{proof}
Let $\mathcal{N}_i(G^{\phi})$ be the set of neighbors of $\eta_i$ in $G^{\phi}$. We have that:
    \begin{align*}
        S_2^*(G^{\phi})
        &   =  \sum_{i=1}^n d_i^-(G^{\phi}) (d_i(G^{\phi}) - 1) \leq \sum_{i=1}^n d_i(G^{\phi}) \times d_i^-(G^{\phi})  =  \sum_{i=1}^n d_i(G^{\phi}) \sum_{\eta_j \in \mathcal{N}_i(G^{\phi})} \mathbbm{1}_{j<i} \\
        &   =  \sum_{(\eta_i,\eta_j) \in E^{\phi}} d_{\max(i,j)}(G^{\phi})
    \end{align*}

Let \(\tau_i\) denote the noise added to the estimated degree of user \(i\). For each edge \((\eta_i, \eta_j)\), their ranks can only be exchanged if the sum of the errors in both degree estimations exceeds the gap between the two degrees. Therefore, the quantity \(d_{\max(i,j)}(G^{\phi})\) satisfies 
\begin{equation*}
 d_{\max(i,j)}(G^{\phi}) \leq \min(d_i, d_j) + |\tau_i| + |\tau_j|.   
\end{equation*}
    Using this inequality, we can rewrite the count of $S_2^*(G^{\phi})$ as
    \begin{equation*}
        S_2^*(G^{\phi}) \leq 
            \sum\limits_{(\eta_i,\eta_j) \in E^{\phi}} \min{(d_i, d_j)}
            + \sum\limits_{i=1}^n |\tau_i| d_i.
    \end{equation*}

    Since $\tau_i$ is sampled from $\Lap(1/\varepsilon_0)$, we have that $|\tau_i|$ follows an exponential law of expectation $1 / \varepsilon_0$. Hence, 
    \begin{equation*}
        \E{S_2^*(G^{\phi})} \leq 
            \sum\limits_{(\nu_i,\nu_j) \in E^{\phi}} \min{(d_i, d_j)}
            + \frac{m}{\varepsilon_0}.
    \end{equation*}

    Since $G$ is isomorphic to $G^\phi$, $\alpha(G) = \alpha(G^{\phi})$ and using Theorem \ref{thm:chiba-nishizeki} it follows that
    \begin{equation*}
      \sum\limits_{(\nu_i,\nu_j) \in E^{\phi}} \min{(d_i, d_j)} \leq m \cdot \alpha(G).  
    \end{equation*}
    
    % Thus, we can introduce $(F_i)_{1 \leq i \leq \alpha(G)}$ the edge-disjoint spanning forest of $G$. Using these notations, the sum can be rewritten as
    %     $\sum\limits_{k=1}^{\alpha(G)} \sum\limits_{(\nu_i,\nu_j) \in E(F_k)} \min{(d_i, d_j)}$.

    %  By choosing a root for each tree of the forests, we can introduce a function $h$ that associates each edge $e$ with its child node. In this way, each node is responsible for at most one edge in each forest. When $m = |E|$, we have
    %  \begin{align*}
    %     \sum_{(\nu_i,\nu_j) \in E^{\phi}} \min{(d_i, d_j)}
    %     \leq \sum_{k=1}^{\alpha(G)} \sum_{e \in E(F_k)} d_{h(e)}(G^{\phi}) \leq \sum_{k=1}^{\alpha(G)} \sum_{i=1}^n d_i(G^{\phi}) = m \cdot \alpha(G)
    % \end{align*}

    Since $m \leq n\delta$ and $\alpha(G) = O(\delta)$, this gives $\E{S_2^*(G^{\phi})} \leq \bigo{\delta^2 n}$.
\end{proof}

In addition to the ordered stars we just discussed, arboricity can also be used to bound the number of paths and cycles in a graph, as demonstrated in the following lemma and theorem. Recall that $\#P_k(G)$ is the number of paths with length $k$ in the graph $G$.

\begin{lemma}
    \label{lem:paths}
    For any positive integer $k$, $\#P_{2k}(G) = \bigo{\delta^{k} n^{k+1}}, \#P_{2k+1} = \bigo{\delta^{k+1} n^{k+1}}$.
\end{lemma}
\begin{proof}
    We first consider \(\#P_{2k + 1}(G)\). Let \(f\) be a function that maps a path of length \(2k + 1\) to a tuple of \(k + 1\) edges, defined as \(f(e_1, \dots, e_{2k + 1}) = (e_1, e_3, \dots, e_{2k + 1})\). We observe that, for any tuple of \(k + 1\) edges denoted by \(\mathcal{E} = (e'_1, \dots, e'_{k + 1})\), \(f^{-1}(\mathcal{E})\) is either a set containing one path or an empty set. There is at most one path that uses \(e'_i\) as the \((2i - 1)\)-th edge of the path for all \(i\). Thus, we can conclude that the number of paths of length \(2k + 1\) is at most the number of sets of \(k + 1\) edges, which is \(m^{k+1} = \bigo{\delta^{k+1} n^{k+1}}\).

    Next, let us consider \(\#P_{2k}(G)\). Let $f$ be a function that maps a path of length $2k$ to a tuple of $k$ edges, defined as $f(e_1, \dots, e_{2k}) = (e_1, e_3, \dots, e_{2k - 1})$. We observe that, for any tuple of $k$ edges denoted by $\mathcal{E} = (e'_1, \dots, e'_k)$, \(f^{-1}(\mathcal{E})\) is a set of size no larger than $n$. There is at most one path of length \(2k - 1\) that uses \(e'_i\) as the \((2i - 1)\)-th edge of the path, and there are at most \(n\) possible ways to extend a path of length \(2k - 1\) to a path of length \(k\). Hence, $\#P_{2k}(G) \leq n \cdot m^k = \bigo{\delta^k n^{k + 1}}$.
\end{proof}

Recall that $\#C_k(G)$ is the number of cycles with size $k$ in the graph $G$. We obtain the following theorem.

\begin{theorem}
    \label{thm:k-cycles}
    For any $k \geq 1$, 
        $\#C_{k+2}(G) \leq \frac{2}{k} \alpha(G) \#P_{k}(G)$.
\end{theorem}

\begin{proof}
%    Let $G = (V, E)$ be a graph and $k \geq 1$.
   Let us denote $\#P_k^{(i)}$ the number of paths of length $k$ that have node $\nu_i$ as an extremity and $\#C_k^{(i,j)}$ the number of cycles of length $k$ containing edge $(\nu_i, \nu_j)$. Using these notations, we have
        $\#C_{k+2} = \frac{1}{k} \sum\limits_{(\nu_i, \nu_j) \in E} \#C_{k + 2}^{(i,j)}$. Consider the number \(\#C_{k + 2}^{(i,j)}\). For a path of length \(k\) that has a node \(\nu_i\) as a terminal, there is at most one cycle of length \(k + 2\) which includes this path and the edge \((\nu_i, \nu_j)\). Therefore, we conclude that \(\#C_{k + 2}^{(i,j)} \leq \#P_k^{(i)}\). Similarly, we have \(\#C_{k + 2}^{(i,j)} \leq \#P_k^{(j)}\). Hence,
        \begin{equation*}
           \#C_{k+2} \leq \frac{1}{k} \sum\limits_{(\nu_i, \nu_j) \in E} \min{\left(\#P_k^{(i)}, \#P_k^{(j)}\right)}. 
        \end{equation*}
    
    For any function $h : E \rightarrow \{1, \dots, n\}$ such that for all $e = (\nu_i, \nu_j) \in E$, $h(e)$ is equal to either $i$ or $j$, $\min{\left(\# P_k^{(i)}, \# P_k^{(j)}\right)} \leq \# P_k^{(h(\nu_i, \nu_j))}$.
    By definition of the arboricity, there exist a set of disjoint forests $\{F_l\}_{l=1, \ldots, \alpha(G)}$ such that $E = \bigcup_{l=1}^{\alpha(G)} F_l$. By choosing a root for each tree of these forests, we can introduce a function $h$ such that each edge has its child node as an image. In this way, each node can only be the image of one edge per forest. This leads to
    \begin{eqnarray*}
        \#C_{k+2}
        & \leq & \frac{1}{k} \sum_{l=1}^{\alpha(G)} \sum_{(\nu_i, \nu_j) \in F_l} \min{\left(\#P_k^{(i)}, \#P_k^{(j)}\right)} \leq \frac{1}{k} \sum_{l=1}^{\alpha(G)} \sum_{e \in F_l} \# P_k^{(h(e))} \leq \frac{1}{k} \sum_{l=1}^{\alpha(G)} \sum_{i \in V} \# P_k^{(i)} \\ &=& \frac{2}{k} \alpha(G) \# P_k.
    \end{eqnarray*}

    The last step is justified by the fact that each path having two extremities, the sum of all the paths of length $k$ starting with node $\nu_i$ is twice the number of paths of length $k$.
\end{proof}

Combining Lemma \ref{lem:paths} and Theorem \ref{thm:k-cycles}, we obtain the following corollary. 
\begin{corollary}
For $k \geq 1$, $\# C_{2k+2} = \bigo{\delta^{k+1} n^{k+1}}$ and $\# C_{2k+1} = \bigo{\delta^{k+1} n^{k}}$. \label{col:numCycles}
\end{corollary}

Next, we focus on the number of cycles of length \(2k\) for any \(k \geq 2\), in which three consecutive vertices of the cycle exhibit monotonic ranks $C^*_{2k}$, as illustrated in Table~\ref{tab:subgraphs}. Throughout the rest of this article, we will denote the count of such subgraphs in \(G\) by \(\#C_{2k}^*(G)\), omitting \(G\) from the notation when the context is clear. In the following theorem, for simplicity, we adopt an abuse of notation by assuming $\#P_{-1}(G) = 1$ and $\#P_{0}(G) = n$ for every graph $G$.

\begin{theorem}
    \label{thm:star-cycles}
    For $k \geq 2$, 
        $\#C_{2k}^*(G) \leq 2 \alpha(G) S_2^*(G) \#P_{2k-5}(G)$.
%    Thus, with $\phi$ one of the ordering of $G$ with respect to the estimated degrees we have
%    \begin{equation}
%        \label{eq:4-cycles-final}
%        C_4^*(G^{\phi}) = \bigo{\delta^{k+1} n^{k-1}}.
%    \end{equation}
\end{theorem}

\begin{proof}
%    Equation~\ref{eq:4-cycles-final} follows immediately from equation~\ref{eq:4-cycles-intermediate} using theorem~\ref{thm:s2-bound} and lemma~\ref{lem:paths}. Thus we will now focus on proving equation~\ref{eq:4-cycles-intermediate}.
Let \(\# C_{2k}^{*(i,j)}(G)\) represent the number of subgraphs in \(G\) where three consecutive vertices exhibit monotonic ranks, with \((\nu_i, \nu_j)\) being the edge immediately following these consecutive vertices.
 Also, for $k \geq 2$, let the number of paths of length \(p\) with a low-2-star as one of its extremities be denoted as \(\#P^*_p\). Since we can construct at most one path included in \(\#P^*_p\) where a low-2-star and a path of length \(p-3\) are its extremities, we obtain the inequality \( \#P^*_p \leq S_2^* \cdot \# P_{p-3} \). 
 
 Let $C_{2k}^{*(i,j)}$ be a cycle which is counted in $\#C_{2k}^{*(i,j)}$. Consider the path in $C_{2k}^{*(i,j)}$ of length \(2k - 2\) starting from \(\nu_i\) that does not pass through \(\nu_j\) and the other path in $C_{2k}^{*(i,j)}$ of the same length starting from \(\nu_j\) that does not pass through \(\nu_i\). We observe that one extremity of the two paths is a low-2-star. Hence, $\# C_{2k}^{*(i,j)} \leq \min{\left(\# P_{2k-2}^{*(i)}, \# P_{2k-2}^{*(j)}\right)}$ when $\# P_p^{*(i)}$ is the number of paths in the count of $\#P^*_{p}$ that have $\nu_i$ as an extremity.
    Using the same definition of $h$ as in the proof of Theorem~\ref{thm:k-cycles}, we have
    \begin{align*}
        \# C_{2k}^*(G)
        & \leq \sum_{(\nu_i, \nu_j) \in E} \# C_{2k}^{*(i,j)}  \leq \sum_{(\nu_i, \nu_j) \in E} \min{\left(\# P_{2k-2}^{*(i)}, \# P_{2k-2}^{*(j)}\right)} \\ & \leq \sum_{l=1}^{\alpha(G)} \sum_{(\nu_i, \nu_j) \in F_l} \min{\left(\# P_{2k-2}^{*(i)}, \# P_{2k-2}^{*(j)}\right)} \leq \sum_{l=1}^{\alpha(G)} \sum_{e \in F_l} \# P_{2k-2}^{*(h(e))} \\ & \leq \sum_{l=1}^{a(G)} \sum_{i \in V} \# P_{2k-2}^{*(i)} \leq 2 \alpha(G) \# P_{2k-2}^{*}  \leq 2 \alpha(G) S_2^{*} \# P_{2k-5}. \qedhere
    \end{align*}
\end{proof}
The next corollary follows Theorem~\ref{thm:s2-bound}, \ref{thm:star-cycles}, and Lemma~\ref{lem:paths}.
\begin{corollary}
    \label{col:numCycle}
    For $k \geq 2$, $\mathbb{E}[C_{2k}^*(G^{\phi})] = \bigo{\delta^{k+1} n^{k-1}}$.
\end{corollary}
The next corollary considers the number of edge sets in $G^\phi$ with specific properties.

\begin{corollary}
    \label{lem:cycles-set}
For any \(p \in \mathbb{N}\), we consider edge sets \(\mathsf{E} \subseteq E^\phi\) of size \(2p\) such that 1) for some $c > 0$, there exists a set of cycles \(C_1, \dots, C_c\) in \(G^\phi\) where \(C_1 \cup \dots \cup C_c = \mathsf{E}\) and \(C_i \cap C_j = \emptyset\) for \(i \neq j\), and 2) at least one of \(C_1, \dots, C_c\) contains three consecutive vertices of monotonic index. The number of such edge sets is \( \bigo{\delta^{p+1} n^{p-1}}\).
\end{corollary}

\begin{proof}
    Consider a partition of \(2p\), denoted by \((p_1, \dots, p_c)\), where \(p_1 + \cdots + p_c = 2p\). The number of such partitions is a function of \(p\) and can be considered constant. We will demonstrate that the number of cycle sets \(C_1, \dots, C_c\) satisfying the conditions in the corollary statement, with \(|C_i| = p_i\), is at most \(\bigo{\delta^{p + 1} n^{p-1}}\). Therefore, the number of cycle sets satisfying the corollary statement is no more than \(\bigo{\delta^{p+1} n^{p-1}}\).
    
    To prove the bound, we will consider two cases: either all the cycles have even lengths, or at least two of them have odd lengths, given that the total number of edges is even.

    If all the cycles are of even length, then, for some $q > 0$ one of them is of length $2q$ and includes 3 consecutive vertices of monotonic index. By Corollary~\ref{col:numCycle}, there are $\bigo{\delta^{q+1} n^{q-1}}$ possibilities for this cycle. For the remaining cycles, Corollary~\ref{col:numCycles} tells us the number of admissible configurations is bounded by $\bigo{\delta^{p-q} n^{p-q}}$. In total, this gives a $\bigo{\delta^{p+1} n^{p-1}}$ bound.
If at least two cycles have odd lengths, say \(2q+1\) and \(2r+1\), then by Corollary~\ref{col:numCycles}, the number of possible configurations for these cycles can be bounded by \(\bigo{\delta^{q+1} n^{q}}\) for the first cycle and \(\bigo{\delta^{r+1} n^{r}}\) for the second cycle, and \(\bigo{\delta^{p-q-r-1} n^{p-q-r-1}}\) for the remaining cycles. Overall, this results in a bound of \(\bigo{\delta^{p+1} n^{p-1}}\).
\end{proof}

\section{Triangle Counting Algorithm}
\label{sec:triangle}

%By expressing the expected error of triangle counting methods in terms of $S_2^*(G)$ and $C_4^*(G)$, preprocessing the graph by reordering it according to the estimated degrees allows us to use theorem~\ref{thm:s2-bound} and \ref{thm:star-cycles} to reduce the error and provide a tighter bound.
We propose Algorithm \ref{algo:triangle-counting} to count the number of triangles based on the ordering and properties discussed in the previous section. First, we execute Algorithm \ref{algo:vertex-ordering} at Line 2. Next, at Line 3, we use the randomized response query to obtain an obfuscated graph. From Lines 4 to 8, we employ the Laplacian query with restricted sensitivity on \(\mathcal{H}_d\) (Definition \ref{def:restricted}) to estimate the number of triangles associated with User \(i\). Finally, at Line 9, we sum all the estimates and report the total as the estimated triangle count. We adopt the concept from \cite{imola2021locally} of distributing randomized response results to all nodes and having each node estimate its number of triangles. However, the other algorithmic ideas presented in this work are novel. In the following theorem, we demonstrate that our algorithm is differentially private.

\begin{algorithm}
    \caption{Our algorithm for estimating the number of triangles in degeneracy-bounded graphs}
    \label{algo:triangle-counting}
    \small
    \Fn{\TriangleCounting}{
        \small
        \KwIn{Graph $G = (V,E)$, privacy budget $\varepsilon = \varepsilon_0 + \varepsilon_1 + \varepsilon_2$, parameter $\zeta$}
        \KwOut{Estimation of the number of triangles in $G$}

        \textbf{[All Users and Server]} $\phi \leftarrow \GetOrdering(G, \varepsilon_0)$\ (Algorithm \ref{algo:vertex-ordering}); \\ 
        \textbf{[All Users and Server]} Inquire the unbiased randomized response query with privacy budget $\varepsilon_1$ to all users. Let \((\hat{a}_{j,k}^\phi)\) represent the results collected from this query. The server then distributes \((\hat{a}_{j,k}^\phi)\) to all users. \\
        
            \textbf{[User $i$]} $\hat{d}^{\phi}_i \leftarrow \tilde{d}^{\phi}_i + \frac{1}{\varepsilon_0} \ln(n / \zeta)$\;
            \textbf{[User $i$]} $a^{\phi}_i \leftarrow \mu_{\hat{d}_i^\phi} (a^{\phi}_i)$ (The function $\mu_d$ is defined before Definition \ref{def:restricted}.) \;
            \textbf{[User $i$]} $S_i \leftarrow \{(j,k) \mid a^{\phi}_{i,j} = a^{\phi}_{i,k} = 1, j < i < k\}$\;
            \textbf{[User $i$]} $\hat{t}_i \leftarrow \sum_{(j,k) \in S_i} \hat{a}^{\phi}_{j,k}$\;
            \textbf{[User $i$]} $\tilde{t}_i \leftarrow \hat{t}_i + 3 \cdot  \Lap(\frac{e^{\varepsilon_1} + 1}{e^{\varepsilon_1} - 1} \cdot \frac{\hat{d}^{\phi}_i}{\varepsilon_2})$\;
        \textbf{[User $i$]} Upload $\tilde{t}_i$ to the central server\;
        \textbf{[Server]} $\hat{f}_{\triangle}(G) \leftarrow \sum_{\nu_i \in V} \tilde{t}_i$\;
        \Return $\hat{f}_{\triangle}(G)$\;
    }
\end{algorithm}

%In algorithm~\ref{algo:triangle-counting}, the function $\GraphProjection(a, k)$ returns the adjacency list $a$ from which only the first $k$ edges in canonical order have been kept.

%Compare to state of the art \cite{imola2021locally} on triangle counting, algorithm~\ref{algo:triangle-counting} leverages the degree clipping techniques described in \cite{blocki2013differentially} as well as the vertex ordering described above.

%The parameter $\alpha$ in the algorithm needs to be carefully set. Indeed, if it is too large, it will make the estimated degrees used larger and thus the noise added via Laplacian mechanism as well. However if it is too small, there is a risk that the estimates could be smaller than the actual degree and thus that some clipping will take place, creating some bias. Lemma~\ref{lem:degree-error} can help set the parameter such that the clipping probability is smaller than a fixed threshold.

\begin{theorem}
    Algorithm~\ref{algo:triangle-counting} provides $(\varepsilon_0 + \varepsilon_1 + \varepsilon_2)$-edge local differential privacy. \label{thm:privacyTriangle}
\end{theorem}

\begin{proof}
    For all possible executions of Algorithm~\ref{algo:triangle-counting}, it inquires three queries to all users. They are 1) the Laplacian query with privacy budget $\varepsilon_0$ inside the \GetOrdering function at Line 2, 2) the unbiased randomized response query with privacy budget $\varepsilon_1$ at Line 3, and 3) the Laplacian query with restricted sensitivity on \(\mathcal{H}_d\) at Lines 4-8. 
    
    To prove this theorem, we only need to show that the query at Lines 4-8 is $\varepsilon_2$-edge local differentially private. The query aims to publish $f(a_i^\phi) = \sum_{(j,k) \in S_i} \hat{a}^\phi_{j,k}$. By the unbiased randomized response in Line 3, we have that, for any $j,k,j',k'$, $|\hat{a}^\phi_{j,k} - \hat{a}^\phi_{j',k'}| \leq \frac{e^{\varepsilon_1} + 1}{e^{\varepsilon_1} - 1}$. 
    It can be shown that, for $a_i^\phi,a_i^{'\phi} \in \mathcal{H}_{\hat{d}_i^\phi}$ (defined in Definition \ref{def:restricted}) such that $d(a_i^\phi,a_i^{'\phi}) \leq \mathsf{d}$, the number of different elements in the set $S_i$ obtained from $a_i^\phi$,$a_i^{'\phi}$ at line 6 is at most $\mathsf{d} \cdot \hat{d}_i^\phi$. Therefore, the restricted sensitivity of the function $f$ (denoted by $RS_f\left(\mathcal{H}_{\hat{d}_i^\phi}\right)$in Definition \ref{def:restricted}) is not larger than  $\mathsf{d} \cdot \hat{d}_i^\phi \cdot \frac{e^{\varepsilon_1} + 1}{e^{\varepsilon_1} - 1} \cdot \frac{1}{\mathsf{d}} = \hat{d}_i^\phi \cdot \frac{e^{\varepsilon_1} + 1}{e^{\varepsilon_1} - 1}$. Hence, by Definition \ref{def:restricted}, the publication of $\tilde{t}_i$ at line 8 is $\varepsilon_2$-edge local differentially private.  
\end{proof}

We now discuss the accuracy of our estimation and its relation with the parameter $\zeta$ appearing at Line 4 the algorithm. We will see that $\zeta$ controls the trade-of between the bias and the accuracy. the smaller $\zeta$ is, the smaller the average noise gets, but the larger the probability of bias and its expected magnitude is.

In the following lemma, we discuss that the projection $\mu_{\hat{d}_i}^{\phi}$ applied at Line 5 changes the adjacency vector $a_i^\phi$ only with small probability.

\begin{lemma}
    \label{lem:degree-error}
For any $\zeta > 0$, with probability at least \(1 - \zeta\),  \(|\tilde{d}_i - d_i| < (\ln{\frac{n}{\zeta}})/{\varepsilon_0}\) for all \(i\).\end{lemma}

\begin{proof}
    Using the cumulative distribution function of the Laplacian random variable, we have 
        $\prob{|\tilde{d}_i - d_i| \geq \varepsilon_0 \ln{\frac{n}{\zeta}}} \leq \frac{\zeta}{n}$.
    Thus, by taking this inequality for all $i \in [1,n]$, and using the union bound, we obtain
        $\prob{\exists i \in [1,n], |\tilde{d}_i - d_i| \geq \varepsilon_0 \ln{\frac{n}{\zeta}}} \leq \zeta$. \qedhere
\end{proof}

We show that our estimation has no bias with high probability in the subsequent theorem.

\begin{theorem}
    %Given that $\hat{d}^{\phi}_i$ is smaller than $d^{\phi}_i$ for all nodes in the graph, algorithm~\ref{algo:triangle-counting} provides an unbiased estimate of the number of triangles in the graph, ie
    With probability at least $1 - \zeta$, algorithm~\ref{algo:triangle-counting} provides an unbiased estimate of the number of triangles in the graph, i.e. 
        $\E{\hat{f}_{\triangle}(G)} = \# C_3(G).$ \label{thm:biasTriangles}
\end{theorem}

\begin{proof}
As discussed in Definition \ref{def:unbiased}, we have that $\mathbb{E}(\hat{a}_{j,k}^\phi) = a_{j,k}^\phi$. Using Lemma~\ref{lem:degree-error}, with probability at least $1-\zeta$, $\hat{d}^{\phi}_i$ is larger than $d^{\phi}_i$ for all $i \in [1,n]$, and the function $\mu_{\hat{d}_i^\phi}$ has no effect. Consequently, $S_i$ precisely represents the set of forks centered on node $\nu_i$, encompassing all possible triangles. Therefore, $\hat{t}_i$ is an unbiased estimate of the number of triangles $(\nu_i, \nu_j, \nu_k)$ such that $j < i < k$. Given that Laplace noise is centered and triangles can be decomposed accordingly, $\hat{f}_{\triangle}(G)$ is an unbiased estimation of $f_{\triangle}(G)$.
\end{proof}

Corollary~\ref{col:bias-triangle} ensures that even in the unlikely event of some clipping occurring, the resulting bias would still represent only a small fraction of the actual count.

\begin{corollary}
    \label{col:bias-triangle}
    The expected value of the bias of Algorithm~\ref{algo:triangle-counting} is bounded by \(\bigo{\frac{\zeta}{\varepsilon_0 n} \# C_3}\).
\end{corollary}

\begin{proof}
When the corrected estimated degree \(\hat{d}^{\phi}_i\) is smaller than the actual degree \(d_i\), \(d_i - \hat{d}^{\phi}_i\) edges are excluded. This exclusion introduces a bias because the potential triangles involving these excluded edges are not counted.
    For each user \(i\) and their neighbor \(j\), let \(t_i^{(j)}\) denote the number of triangles counted by user \(i\) that involve the edge \((\nu_i, \nu_j)\). We also define \(t_i^{\text{max}} = \max_j t_i^{(j)}\). Then, the maximum bias resulting from a single clipped edge can be bounded by \(t_i^{\text{max}}\).

   The expected number of clipped edges for user \(i\) is determined by evaluating the following integral, where \(\beta = \frac{\ln(n/\zeta)}{\varepsilon_0}\) serves as the correction term for the degree:
    \[
        \int_{-\infty}^{-\beta} \frac{\varepsilon_0}{2} e^{-\varepsilon_0 |x|} (-x - \beta) \,dx = 
        - \frac{\varepsilon_0}{2}\left[\frac{x - \beta}{\varepsilon_0} e^{-\varepsilon_0 x} + \frac{1}{\varepsilon_0^2} e^{-\varepsilon_0 x}\right]_{\beta}^{\infty} = 
        \frac{1}{2 \varepsilon_0} e^{-\varepsilon_0 \beta} = 
        \frac{\zeta}{2 \varepsilon_0 n}.
    \]

We obtain the final result by combining these elements and observing that \(\sum_i t_i^{\text{max}} \leq \sum_{i,j} t_i^{(j)} \leq 2 f_{\triangle}(G)\).
\end{proof}

The accuracy of our estimation is demonstrated in the subsequent theorem.

\begin{theorem}
    When $\zeta \leq \varepsilon_0$, the squared expected $\ell_2$-error of algorithm~\ref{algo:triangle-counting} is bounded by
    \[
        \bigo{
            \frac{\delta^3 n}{\varepsilon_1^2} + 
            \frac{\delta d_{max} n }{\varepsilon_1^2 \varepsilon_2^2} + 
            \frac{n \ln^2(n / \zeta)}{\varepsilon_0^2 \varepsilon_1^2 \varepsilon_2^2}
        }.
    \]
\end{theorem}

\begin{proof}
    The squared $\ell_2$-error can be decomposed into the square of the bias plus the variance. We have established in Corollary~\ref{col:bias-triangle} that the bias of the algorithm is bounded by $\bigo{\frac{\zeta}{\varepsilon_0 n} \#C_3} = \bigo{\delta^2}$.
    We will now focus on bounding the variance of the algorithm.
    %We have previously established that the algorithm provides an unbiased estimate with a probability of at least $1 - \zeta$. In this case, the squared $\ell_2$-error is equivalent to the variance of the estimate. Therefore, we will now focus on bounding the variance of the algorithm.
    This variance arises from two distinct sources: the randomized response query and the Laplacian query with restrictive sensitivity.
    
    Regarding the noise introduced by the Laplacian query with restrictive sensitivity, its variance is simply the sum of the variances of each term, which is
    \begin{eqnarray*}
     9 \left(\frac{e^{\varepsilon_1} + 1}{e^{\varepsilon_1} - 1}\right)^2 \sum_{\nu_i \in V} \frac{\hat{d}_i^2}{\varepsilon_2^2} = 
    \bigo{\frac{\varepsilon_0^2 \delta d_{max} n + n \ln^2(n / \zeta)}{\varepsilon_2^2 \varepsilon_1^2 \varepsilon_0^2}}.   
    \end{eqnarray*}

    Next, we consider the variance from the randomized response query. In the following equations, we use the notation $\mathcal{N}^*_{j,k}$ to denote the set of neighbors $\nu_i$ of both $\nu_j$ and $\nu_k$ such that $j < i < k$. Note that by including one node from $\mathcal{N}^*_{j,k}$ along with $\nu_j$ and $\nu_k$, a triple in $S^{*}_2$ is formed. Similarly, including two nodes from $\mathcal{N}^*_{j,k}$ along with $\nu_j$ and $\nu_k$ results in a quadruplet in $\#C^*_4$. We also notice from Definition \ref{def:unbiased} that, for $(j,k) \neq (j',k')$, $\hat{a}_{j,k}^\phi$ is independent to $\hat{a}_{j',k'}^\phi$ and $\cov{\hat{a}_{j,k}^\phi}{\hat{a}_{j',k'}^\phi} = 0$. Hence, 
    \begin{align*}
        \var{\sum_{\nu_i \in V^{\phi}} \sum_{(j,k) \in S_i} \hat{a}^{\phi}_{j,k}} & =  \sum_{(\nu_j, \nu_k) \in (V^{\phi})^2} 
            \left[
                \sum_{\nu_i \in \mathcal{N}^*_{j,k}} \var{\hat{a}^{\phi}_{j,k}}
                + \sum_{\nu_i, \nu_{i'} \in \mathcal{N}^*_{j,k}} \cov{\hat{a}^{\phi}_{j,k}}{\hat{a}^{\phi}_{j,k}}
            \right] \\
        & = \bigo{(S^*_2 + \# C^*_4)/\varepsilon^2_1} 
    \end{align*}
By Theorem \ref{thm:s2-bound} and Collorary \ref{col:numCycle}, 
        $\var{\hat{f}(G)} = 
        \bigo{
            \frac{\delta^3 n}{\varepsilon_1^2} + 
            \frac{\delta d_{max} n }{\varepsilon_1^2 \varepsilon_2^2} + 
            \frac{n \ln^2(n / \zeta)}{\varepsilon_0^2 \varepsilon_1^2 \varepsilon_2^2}
        }$.
        \qedhere
 \end{proof}

In the previous work \cite{eden2023triangle}, the number of terms in the variance calculation is bounded by the number of cycles of length four, which is $\bigo{{d^3_{\textrm{max}}} n}$. We reduce that number to $\#C_4^* = \bigo{\delta^3n}$ using the \GetOrdering function in line 2 and by including only pairs $(j,k)$ such that $j < i < k$. It is known that $\delta \leq d_{\max}$ and, in many practical graphs, the degeneracy is much smaller than the maximum degree.

\section{Odd Length Cycle Counting}
\label{sec:cycles}

In this section, we will describe how to utilize low-degree ordering to accurately count odd-length cycles in graphs with bounded degeneracy. Some concepts are extended from the previous section. As shown in Algorithm \ref{algo:cycle-counting}, the algorithm for estimating the number of odd-length cycles is similar to Algorithm \ref{algo:triangle-counting}, except that the restricted sensitivity at Line 9 is larger, and at Line 8, we replace $\hat{a}^\phi_{i,j}$ with an estimate for the number of paths under specific constraints. We discuss the privacy of the algorithm in the subsequent theorem. The main challenge of the proof is to demonstrate that the Laplacian query under restricted sensitivity at Lines 5-9 is $\varepsilon_2$-differentially private. 

\begin{algorithm}[h]
    \caption{  Our algorithm for estimating the number of odd-length cycles in degeneracy-bounded graphs}
    \label{algo:cycle-counting}
    \small
    \Fn{\OddCycleCounting}{
        \KwIn{Graph $G = (V,E)$, privacy budget $\varepsilon = \varepsilon_0 + \varepsilon_1 + \varepsilon_2$, $k$ an odd number not smaller than 5, parameter $\zeta$}
        \KwOut{Estimation of the number of $k$-cycles in $G$}

        \textbf{[All Users and Server]} $\phi \leftarrow \GetOrdering(G, \varepsilon_0)$\ (Algorithm \ref{algo:vertex-ordering}); \\ 
        \textbf{[All Users and Server]} Inquire the unbiased randomized response query with privacy budget $\varepsilon_1$ to all users. \\
        \textbf{[Server]} Let \((\hat{a}_{i,j}^\phi)\) represent the results collected from this query. The server then distributes \((\hat{a}_{i,j}^\phi)\) to all users. \\
        \textbf{[Server]} Calculate \begin{equation*}
            \# \hat{P}_{k - 4} := \sum\limits_{(l_1, \dots, l_{k - 3}) \in V^{k - 3}} \prod\limits_{q \in [1, k-4]} \hat{a}^{\phi}_{l_q, l_{q+1}},
        \end{equation*} then send this information to all users\;
        
        \textbf{[User $i$]}    $\hat{d}^{\phi}_i \leftarrow \tilde{d}^{\phi}_i + \frac{1}{\varepsilon_0} \ln(n / \zeta)$\;
            
        \textbf{[User $i$]} $a^{\phi}_i \leftarrow \mu_{\hat{d}^{\phi}_i} (a^{\phi}_i)$\;
            
            \textbf{[User $i$]} $S_i \leftarrow \{(j,\kappa) \mid a^{\phi}_{i,j} = a^{\phi}_{i,\kappa} = 1, j < i < \kappa\}$\;
            
            \textbf{[User $i$]} $\hat{c}_i \leftarrow \sum_{(j,\kappa) \in S_i} \# \hat{P}_{k-2}^{(i)}(j,\kappa)$ when
            \begin{equation*}
             \# \hat{P}_{k-2}^{(i)}(j,\kappa) = \sum\limits_{(l_1, \ldots, l_{k-1}) \in X_{k-2}^{(i)}(j,\kappa)} \prod\limits_{q \in [1, k-2]} \hat{a}^{\phi}_{l_q, l_{q+1}}   
            \end{equation*}
             and $X_{k-2}^{(i)}(j,\kappa)$ is a set of non-repeating combination of $k-1$ vertices in $G^{\phi}$ with endpoints $\nu_j$ and $\nu_\kappa$, such that, for any three consecutive nodes $(\nu_q, \nu_r, \nu_s)$ in the path with monotonic ranks, the node $\nu_i$ has a lower rank than $\nu_r$\;
            
            \textbf{[User $i$]} $\tilde{c}_i \leftarrow \hat{c}_i + \Lap\left(3 \cdot \left(\frac{e^{\varepsilon_1} + 1}{e^{\varepsilon_1} - 1}\right)^2 \cdot \hat{d}^{\phi}_i \cdot \#\hat{P}_{k-4} / \varepsilon_2\right)s$\;
            \textbf{[User $i$]} Upload $\tilde{c}_i$ to the central server\;
        \textbf{[Server]} $\hat{f}_k(G) \leftarrow \sum\limits_{\nu_i \in V} \tilde{c}_i$\;
        \Return $\hat{f}_k(G)$\;
    }
\end{algorithm}

\begin{theorem}
    Algorithm~\ref{algo:cycle-counting} provides $(\varepsilon_0 + \varepsilon_1 + \varepsilon_2)$-edge local differential privacy. \label{thm:privacyCycles}
\end{theorem}

\begin{proof}
We need to demonstrate that Lines 5-9 of the algorithm, involving the Laplacian query with restricted sensitivity on $\mathcal{H}_{\hat{d}_i^\phi}$, ensure $\varepsilon_2$-differential privacy. Following the arguments of Theorem \ref{thm:privacyTriangle}, we assert that altering $\mathsf{d}$ entries of $a_{i,j}^\phi$ changes the set $S_i$ by at most $\mathsf{d} \cdot \hat{d}_i^\phi$ elements. A single element change in $S_i$ can alter the value of $\hat{c}_i$ by 
\begin{equation*}
 \# \hat{P}_{k-2}^{(i)}(j,\kappa) = \sum\limits_{(l_1, \ldots, l_{k-1}) \in X_{k-2}^{(i)}(j,\kappa)} \prod\limits_{q \in [1, k-2]} \hat{a}^{\phi}_{l_q, l_{q+1}} \leq \left( \frac{e^\varepsilon + 1}{e^\varepsilon - 1} \right)^2 \#\hat{P}_{k - 4}.   
\end{equation*}
Therefore, the restricted sensitivity of $\tilde{c}_i$ is $\mathsf{d} \cdot \hat{d}_i^\phi \left( \frac{e^\varepsilon + 1}{e^\varepsilon - 1} \right)^2 \#\hat{P}_{k - 4} / \mathsf{d} = \hat{d}_i^\phi \left( \frac{e^\varepsilon + 1}{e^\varepsilon - 1} \right)^2 \#\hat{P}_{k - 4}$. Consequently, the publication of $\tilde{c}_i$ at Line 9 is $\varepsilon_2$-differentially private.
\end{proof}

The bias of the algorithm is given in the following theorem. %The proof of this theorem uses similar arguments as Theorem \ref{thm:biasTriangles}.

\begin{theorem}
With a probability of at least \(1 - \zeta\), Algorithm~\ref{algo:cycle-counting} provides an unbiased estimate of the number of \(k\)-cycles in any graph \(G\) for any odd integer \(k\). \label{thm:biasCycles}
\end{theorem}

\begin{proof}
    Since we publish $\hat{a}^{\phi}_{l_q, l_{q+1}}$ using the unbiased randomized response query, the publication is an unbiased estimation of $a^{\phi}_{l_q, l_{q+1}}$.
    Furthermore, as those estimators are independent from one another, for each $(j,\kappa) \in S_i$ and $\{l_1, \ldots, l_{k-1}\} \in X_{k-2}^{(i)}(j,\kappa)$, $\prod\limits_{q \in [1, k-2]} \hat{a}^{\phi}_{l_q, l_{q+1}}$ is an unbiased estimate of $\prod\limits_{q \in [1, k-2]} a^{\phi}_{l_q, l_{q+1}}$.
    It results from this that $\# \hat{P}_{k-2}^{(i)}(j,\kappa)$ is an unbiased estimator of the number of paths between $j$ and $\kappa$ with length $k - 2$ such that, for any three consecutive nodes $(\nu_q, \nu_r, \nu_s)$ with monotonic ranks, the node $\nu_i$ has a lower rank that $\nu_r$. We denote the number of such paths as $\# P_{k-2}^{(i)}(j,\kappa)$.

Let us introduce \( C_k^{(i)} = \sum_{(j,k) \in S_i} \# P_{k-2}^{(i)}(j,\kappa) \). Assuming no clipping occurs, which happens with a probability of at least \( 1 - \zeta \), we have by linearity of expectation that both \(\hat{c}_i\) and \(\tilde{c}_i\) are unbiased estimators of \( C_k^{(i)} \). Therefore, all that remains to be proven is that the number of \(k\)-cycles in \(G\) is equal to \(\sum_{\nu_i \in V^{\phi}} C_k^{(i)}\).
    It is evident that for each element counted in \(\sum_{\nu_i \in V^{\phi}} C_k^{(i)}\), there is a corresponding cycle \((\nu_i, l_1, \ldots, l_{k-1})\) in \(G^{\phi}\) and, also, in \(G\).

Conversely, consider a cycle of length \(k\) in \(G\). Since it is also a cycle in \(G^{\phi}\), we can represent it in \(G^{\phi}\) as \((\nu_1, \ldots, \nu_k)\). Because the cycle is of odd length, there exist three consecutive nodes with a monotonic rank. Among all possible triplets, consider the one where the central node has the smallest rank, denoted as \((\nu_j, \nu_i, \nu_\kappa)\) with \(j < i < \kappa\). Furthermore, let \(j = l_1\) and \(\kappa = l_{k-1}\), and assign the indices of the other nodes in the cycle to \(l_2\) through \(l_{k-2}\) in the order they appear in the cycle. Thus, the cycle is counted in \(\sum_{\nu_i \in V^{\phi}} C_k^{(i)}\).
    Furthermore, if any other node in the cycle were chosen as $\nu_i$, the remaining path would not be part of \(X_{k-2}^{(i)}(j,\kappa)\). This ensures that each cycle is counted exactly once in \(\sum_{\nu_i \in V^{\phi}} C_k^{(i)}\).
    %This proves that the number of $k$-cycles in $G$ is equal to \(\sum_{\nu_i \in V^{\phi}} C_k^{(i)}\).
\end{proof}

Finally, the $\ell_2$-error of Algorithm \ref{algo:cycle-counting} is in the next theorem. The most challenging aspect of this theorem is to bound the covariance in the summation at Lines 8 and 10. We assert that any two dependent elements of \(X_{k - 2}^{(i)}(j, \kappa)\) can be considered as a set containing an even number of edges which forms multiple disjoint cycles with specific properties. Consequently, we can utilize our results from Corollary \ref{lem:cycles-set} to bound the number of such pairs. The proof of the theorem is given in the appendix of this paper.

\begin{theorem}
    \label{thm:cycles-accuracy}
    When $\zeta \leq \varepsilon_0$, the expected squared $\ell_2$-error of algorithm~\ref{algo:cycle-counting} is bounded by
    \[
        \bigo{
            \frac{\delta^3}{\varepsilon_1^2} \left(\frac{1}{\varepsilon_1^2} + \delta\right)^{k-3} n^{k-2} + 
            \frac{\delta^{k-2} d_{max} n^{k-2}}{\varepsilon_2^2 \varepsilon_1^4} + 
            \frac{\delta^{k-3} n^{k-2} \ln^2(n / \zeta)}{\varepsilon_2^2 \varepsilon_1^4 \varepsilon_0^2}
        }.
    \]
\end{theorem}

\section{Conclusion}
\label{sec:conclusion}

In this work, we introduced a private vertex ordering algorithm. The transformation on the graph induced by this ordering reduces the count of specific order-sensitive motifs while preserving the overall graph structure. Due to its reliance on the Laplacian mechanism, the algorithm performs well even in high-privacy settings, making it an excellent preprocessing step for subgraph counting queries.

Within this framework, we first propose a new triangle counting algorithm whose accuracy depends on the count of specific ordered subgraphs. By combining this algorithm with the ordering preprocessing step, we achieve an expected error of \(\bigo{n}\) for graphs with bounded degeneracy, compared to the \(\bigo{n^2}\) error seen in the current state of the art.

Subsequently, we extended the algorithm to address the more general case of odd-length cycle counting. We propose the first purely local differentially private counting algorithm for cycles longer than triangles. Under the assumption of bounded degeneracy, the algorithm achieves an error of \(\bigo{n^{(k-1)/2}}\) for cycles of length \(k\).

Due to the constraints of local differential privacy, it might be assumed that the range of tasks we can perform on graphs under this privacy notion is limited. However, in this work, we demonstrate that more precise information can be published under local differential privacy by restricting our inputs to certain types of graphs. We believe that parameterized algorithms under local differential privacy represent an intriguing research area that can contribute significantly to both algorithm design and information privacy.

%In this work we have first described a private vertex ordering method that performs on a low privacy budget. After establishing some properties on the obtained ordering, we have shown how the ordering can be used as a pre-processing step for a triangle counting algorithm as well as a odd-cycle counting algorithm.
%The resulting mechanism effectively leverages the degeneracy bounded nature of the real world graphs, with the main term of error being $\bigo{\delta(G)^3 n}$ compared to $\bigo{d_{max}^3 n}$ for the state-of-the-art \cite{imola2021locally}.

\newpage

\bibliography{references}

\newpage

\section*{Appendix}

We prove Theorem \ref{thm:cycles-accuracy} in this Appendix. Before that, we demonstrate the following lemma.

\begin{lemma}
    \label{col:bias-cycle}
    The expected value of the bias of Algorithm~\ref{algo:cycle-counting} is bounded by \(\bigo{\frac{\zeta}{\varepsilon_0 n} \# C_k}\).
\end{lemma}

\begin{proof}
    We have already seen the proof of Corollary \ref{col:bias-triangle} that the expected value of the number of clipped edges for user $i$ was bounded by $\frac{\zeta}{2 \varepsilon_0 n}$.
    We now have to bound the bias created by one edge removal, i.e. the maximal number of cycles one edge can part of.

    With $c_i^{(j)}$ the number of cycles counted by $i$ that involve edge $(i,j)$, the maximal bias for user $i$ is bounded by $\sum_j c_i^{(j)}$, and the bias of the algorithm by $\frac{\zeta}{2 \varepsilon_0 n} \sum_{i,j} c_i^{(j)} \leq \bigo{\frac{\zeta}{\varepsilon_0 n} \# C_k}$.
\end{proof}

Now, we are ready to prove Theorem \ref{thm:cycles-accuracy}.

\begin{proof}[Proof of Theorem \ref{thm:cycles-accuracy}] The squared \(\ell_2\)-error can be decomposed into the square of the bias plus the variance.
In Lemma~\ref{col:bias-cycle}, we established that the bias of the algorithm is bounded by \(\bigo{\frac{\zeta}{\varepsilon_0 n} \#C_k} = \bigo{\delta^{\frac{k+1}{2}} n^{\frac{k-3}{2}}}\). We will now focus on bounding the variance of the algorithm.
     
    Let the indicator variable \(\mathbbm{1}_{(l_1, \ldots, l_p)}\) be 1 if the path \((\nu_{l_1}, \ldots, \nu_{l_p})\) exists in \(G^{\phi}\), and 0 otherwise. We also denote the random variable \(\prod_{q \in [1, p]} \hat{a}^{\phi}_{l_q, l_{q+1}}\) by \(Z_{(l_1, \ldots, l_{p+1})}\). Finally, we define \(U_{(l_1, \ldots, l_{p+1})} = Z_{(l_1, \ldots, l_{p+1})} - \mathbbm{1}_{(l_1, \ldots, l_{p+1})}\). This random variable \(U_{(l_1, \ldots, l_{p+1})}\) has the properties that \(\E{U_{(l_1, \ldots, l_{p+1})}} = 0\) and \(\var{U_{(l_1, \ldots, l_{p+1})}} = \var{Z_{(l_1, \ldots, l_{p+1})}}\).

    Similar to the case with triangles, the variance of Algorithm~\ref{algo:cycle-counting} arises from both the unbiased randomized response query and the Laplacian query with restricted sensitivity.

    Concerning the variance term coming from the randomized response, we have to compute the variance of 
    \begin{equation*}
    \hat{C} = \sum\limits_{\nu_i \in V} (\hat{c_i} - c_i) = \sum\limits_{\nu_i \in V} \sum\limits_{(j,\kappa) \in S_i} \sum\limits_{\{l_1, \ldots, l_{k-1}\} \in X_{k-2}^{(i)}(j,\kappa)} U_{(l_1, \ldots, l_{k-1})}.    
    \end{equation*}
    We have to take into account the term that comes from the sum of the variances of the $U$ as well as the one coming from the covariances between them.

To compute the sum of variances, we start with:
\[
\var{U_{(l_1, \ldots, l_{k-1})}} = \prod_{q \in [1, k-2]} \var{\hat{a}^{\phi}_{l_q, l_{q+1}}} = \bigo{\frac{1}{\varepsilon_1^{2k-4}}}.
\]
Additionally, for each \(i\) and \((j,\kappa) \in S_i\), the cardinality of \(X_{k-2}^{(i)}(j,\kappa)\) is bounded by \(n^{k-3}\), and the number of ways to choose \((i,j,\kappa)\) is bounded by \(S_2^*\), which is $\bigo{\delta^2 n}$ by Theorem \ref{thm:s2-bound}. This contributes a term in the variance from the sum of variances bounded by \(\bigo{\delta^2 n^{k-2} / \varepsilon_1^{2k-4}}\).

    To analyze the term arising from the covariances, we first examine the covariance between \(U_{(l_1, \ldots, l_{k - 1})}\) and \(U_{(l'_1, \ldots, l'_{k - 1})}\). In the following equations, let \(A\) be the set of edges that appear only in \((l_1, \ldots, l_{k - 1})\) or \((l'_1, \ldots, l'_{k - 1})\), and let \(B\) be the set of edges that appear in both. Recall that, for any $(i,j)$ $\E{\hat{a}^{\phi}_{i,j}} = 0$ and $\E{\hat{a}^{\phi}_{i,j}} = \var{\hat{a}^{\phi}_{i,j}}$.
    \begin{align*}
        & \cov{U_{(l_1, \ldots, l_{k - 1})}}{U_{(l'_1, \ldots, l'_{k - 1})}} \\
        & = \E{\prod_{q \in [1, k - 2]} \hat{a}^{\phi}_{l_q, l_{q+1}} \prod_{q \in [1, k - 2]} \hat{a}^{\phi}_{l'_q, l'_{q+1}}} - \mathbbm{1}_{(l_1, \ldots, l_{k - 1})} \mathbbm{1}_{(l'_1, \ldots, l'_{k - 1})} \\
        & = \prod_{(i,j) \in A} \mathbbm{1}_{(i,j)} \prod_{(i,j) \in B} \var{a^{\phi}_{i,j}} - \prod_{q \in [1, k - 2]} \mathbbm{1}_{(l_q, l_{q+1})} \mathbbm{1}_{(l'_q, l'_{q+1})}.
    \end{align*}
We observe that the covariance between \(U_{(l_1, \ldots, l_{k - 1})}\) and \(U_{(l'_1, \ldots, l'_{k - 1})}\) is zero if the paths \((\nu_{l_1}, \ldots, \nu_{l_{k - 1}})\) and \((\nu_{l'_1}, \ldots, \nu_{l'_{k - 1}})\) do not share at least one common edge or if the edges present in only one of the paths are not present in the original graph. Now consider the situation where the covariance is non-zero. We have that \(|B| > 0\).
Additionally, we will denote $\nu_i$ the node responsible for counting this instance of \(U_{(l_1, \ldots, l_{k - 1})}\) and $\nu_{i'}$ the one responsible for \(U_{(l'_1, \ldots, l'_{k - 1})}\).

Let \(\mathsf{V} := \{\nu_i, \nu_{l_1}, \dots, \nu_{l_{k - 1}}, \nu_{i'}, \nu_{l'_1}, \dots, \nu_{l'_{k - 1}}\}\), and let
%\[
%    \mathsf{E} := \{ (\nu_{l_q}, \nu_{l_{q+1}}) : 1 \leq q \leq k-2 \} \cup \{ (\nu_{l'_q}, \nu_{l'_{q+1}}) : 1 \leq q \leq k-2 \} \cup \{ (\nu_{l_{k - 1}}, \nu_{l_1}), (\nu_{l'_{k - 1}}, \nu_{l'_1}) \}.
%\]
\[
    \mathsf{E} := \bigcup_{1 \leq q \leq k-2} \{ (\nu_{l_q}, \nu_{l_{q+1}}), (\nu_{l'_q}, \nu_{l'_{q+1}}) \} \cup
    \{ (\nu_{l_{k - 1}}, \nu_i), (\nu_i, \nu_{l_1}), (\nu_{l'_{k - 1}}, \nu_{i'}), (\nu_i, \nu_{l'_1}) \}.
\]
%In other words, the set \(\mathsf{E}\) consists of the edges in the paths \(\{l_1, \ldots, l_{k - 1}\}\) and \(\{l'_1, \ldots, l'_{k - 1}\}\), along with the additional edges \((\nu_{l_{k - 1}}, \nu_{l_1})\) and \((\nu_{l'_{k - 1}}, \nu_{l'_1})\). Additionally, let \(A' := A \cup \{(\nu_{l_{k - 1}}, \nu_{l_1}), (\nu_{l'_{k - 1}}, \nu_{l'_1})\}\) if \((\nu_{l_{k - 1}}, \nu_{l_1}) \neq (\nu_{l'_{k - 1}}, \nu_{l'_1})\), and \(A' := A\) otherwise. Similarly, let \(B' := B \cup \{(\nu_{l_{k - 1}}, \nu_{l_1})\}\) if \((\nu_{l_{k-1}}, \nu_{l_1}) = (\nu_{l'_{k - 1}}, \nu_{l'_1})\), and \(B' := B\) otherwise. In other words, the sets \(A'\) and \(B'\) are the sets \(A\) and \(B\) extended to include the additional edges \((\nu_{l_{k - 1}}, \nu_{l_1})\) and \((\nu_{l'_{k - 1}}, \nu_{l'_1})\). We have that $|A'| \leq 2p$.
In other words, the set \(\mathsf{E}\) consists of the edges in the paths \(\{l_1, \ldots, l_{k - 1}\}\) and \(\{l'_1, \ldots, l'_{k - 1}\}\), along with the additional edges \((\nu_{l_{k - 1}}, \nu_i)\), \((\nu_i, \nu_{l_1})\), \((\nu_{l'_{k - 1}}, \nu_{i'})\), and \((\nu_i, \nu_{l'_1})\).
Additionally, let 
\begin{eqnarray*}
    A' & := A \cup \{(\nu_{l_{k - 1}}, \nu_i), (\nu_{l'_{k - 1}}, \nu_{i'}) \mid (\nu_{l_{k - 1}}, \nu_i) \neq (\nu_{l'_{k - 1}}, \nu_{i'}) \} \\ & \cup \{(\nu_i, \nu_{l_1}), (\nu_{i'}, \nu_{l'_1}) \mid (\nu_i, \nu_{l_1}) \neq (\nu_{i'}, \nu_{l'_1}) \}.
\end{eqnarray*} Similarly, let 
\begin{eqnarray*}
B' & := B \cup \{(\nu_{l_{k - 1}}, \nu_i) \mid (\nu_{l_{k-1}}, \nu_i) = (\nu_{l'_{k - 1}}, \nu_{i'})\} \cup \{(\nu_i, \nu_{l_1}) \mid (\nu_i, \nu_{l_1}) = (\nu_{i'}, \nu_{l'_1})\}.    
\end{eqnarray*}
In other words, the sets \(A'\) and \(B'\) are the sets \(A\) and \(B\) extended to include the additional edges \((\nu_i, \nu_{l_1})\), \((\nu_{l_{k - 1}}, \nu_i)\), \((\nu_{i'}, \nu_{l'_1})\), and \((\nu_{l'_{k - 1}}, \nu_{i'})\).

We introduce $\mathbf{d}$ the difference between the cardinal of $B'$ and $B$, $\mathbf{d} := |B'| - |B|$. Let $q \in [1, k-2]$ be the cardinality of $B$.
In this case, the covariance is \(\bigo{1 / \varepsilon_1^{2q}}\).
We have that $|A'| + 2|B'| = 2 k$, which gives \(|A'| = 2k -2q -2\mathbf{d}\).

In the next step, we will calculate the number of the pairs of paths with $|A'| = 2(k-q-\mathbf{d})$.
%In this case, \(A \neq \emptyset\) and the covariance is \(\bigo{1 / \varepsilon_1^{2r}}\). Let $|A| = 2p - 2$ and $|B| = k - 1 - p$. In the next step, we will calculate the number of the pairs of paths with $|A| = 2p - 2$.
Let us consider the degree of each node in $(\mathsf{V}, \mathsf{E})$. It is clear that the degrees are neither greater than four nor less than two. A node has a degree of three only if one of the three edges incident to it belongs to $B'$ and the other two to $A'$. A node has a degree of four if all four edges incident to it are in $A'$, and it has a degree of two if both edges incident to it are either in $A'$ or in $B'$. Hence, if we consider the graph $(\mathsf{V}, A')$, we have a graph of degree two or four, which is a union of multiple disjoint cycles. 

Let the number of those disjoint cycles be $c$ and the size of those cycles be $r_1, \dots r_c$. We have that $\sum_{t = 1}^{c} r_t = 2k - 2q - 2\mathbf{d}$, i.e. $(r_1, \dots, r_c)$ is a partition of $2k - 2q - 2\mathbf{d}$. We know that the number of such partitions is bounded by a function of $k$. Let suppose that the bound is $f(k)$.

Let us give the number of $A'$ with cycle size $(r_1, \dots, r_c)$. We can use Corollary~\ref{col:numCycles} to show that the number of such sets \(A'\) is \(\bigo{\prod_{t = 1}^c \delta^{r_t / 2} n^{r_t/2}} = \bigo{\delta^{k - q - \mathbf{d}} n^{k - q - \mathbf{d}}}\). When \(\mathbf{d} = 0\), we know that \(\{\nu_i, \nu_j\}\) and \(\{\nu_i, \nu_k\}\) are in \(A'\). There are three consecutive nodes with monotonic ranks in the union of disjoint cycles \((\mathsf{V}, A')\). Hence, we can use Corollary~\ref{lem:cycles-set} to show that the number of such sets \(A'\) is bounded by \(\bigo{\delta^{k - q + 1} n^{k - q - 1}}\). By combining the two cases, we can conclude that the number of possible sets \(A'\) with cycle size $(r_1, \dots, r_c)$ is at most \(\bigo{\delta^{k - q + 1 - \mathbf{d}} n^{k - q - 1}}\). The number of possible $A'$ is then \(f(k) \cdot \bigo{\delta^{k - q + 1 - \mathbf{d}} n^{k - q - 1}}\). As $k$ is a constant, the number is $\bigo{\delta^{k - q + 1 - \mathbf{d}} n^{k - q - 1}}$.  

%We then can use Corollary \ref{col:numCycles} and \ref{lem:cycles-set} to show that the number of such subgraphs is no more than $\bigo{\delta^{k - q + 1 - 2\mathbf{d}} n^{k - q - 1}}$ as at least one cycle has three ordered vertices when $\mathbf{d} = 0$.
%We can then conclude that there is at most $\bigo{\delta^{k - q + 1 - 2\mathbf{d}} n^{k - q - 1}}$ possible configurations for $A'$. 

We then consider the number of configurations for \(B'\), which consists of a union of disjoint paths. Let the number of paths be \(c\) and their lengths be \(r_1, \dots, r_c\). We have that \(|r_1| + \dots + |r_c| = q\), and \((r_1, \dots, r_c)\) forms a partition of \(q\). The number of possible partitions is bounded by a function of \(k\), denoted as \(f(k)\). Each part must begin and end in the node set \(A'\), where \(|A'| \leq 2k\). Therefore, the number of possible paths \(r_t\) is at most \(4k^2 n^{r_t - 1}\), and the number of possible sets \(B'\) with the partition \((r_1, \dots, r_c)\) is at most \(\prod_{t = 1}^{c} 4k^2 n^{r_t - 1} = \bigo{n^{q - 1}}\). Hence, the total number of possible sets \(B'\) is \(f(k) \cdot \bigo{n^{q - 1}} = \bigo{n^{q - 1}}\).

Consequently, for each set \(A'\), the number of possible configurations for \(B'\) is at most \(\bigo{n^{q-1}}\). The number of pairs of paths \(\{l_1, \dots, l_{k-1}\}\) and \(\{l'_1, \dots, l'_{k-1}\}\) with \(|A'| = 2(k-q-\mathbf{d})\) is then at most \(\bigo{\delta^{k - q + 1 - \mathbf{d}} n^{k - q - 1} \cdot n^{q-1}} = \bigo{\delta^{k - q + 1} n^{k - 2}}\). Each of these pairs contributes \(\var{\hat{a}^{\phi}_{j,k}}^{2q} = \bigo{1 / \varepsilon_1^{2q}}\) to the covariance sum.

    The covariance of $\hat{C}$ can then be calculated as follows:
    \[
        \bigo{\sum_{q=1}^{k-2} \delta^{k - q + 1} n^{k - 2} \frac{1}{\varepsilon_1^{2q}}} = \bigo{\frac{n^{k-2}\delta^3}{\varepsilon_1^2} \left(\delta + \frac{1}{\varepsilon_1^2}\right)^{k-3}}.
    \]
Since this bound is larger than the one for the sum of variances, we can disregard the latter.

    To compute the variance resulting from the Laplacian query with restricted sensitivity, we sum the variance of the Laplacian distribution for all nodes:
    \begin{equation}
        9 \left(\frac{e^{\varepsilon} + 1}{e^{\varepsilon} - 1}\right)^4 \E{\# \hat{P}_{k-4}^2} \sum_{\nu_i \in V} \frac{\hat{d}_i^2}{\varepsilon_2^2} = 
        \bigo{\frac{\delta d_{max} n}{\varepsilon_2^2 \varepsilon_1^4} + \frac{n \ln^2(n / \zeta)}{\varepsilon_2^2 \varepsilon_1^4 \varepsilon_0^2}} \E{\# \hat{P}_{k-4}^2}. \label{eqn1}
    \end{equation}
Let us now consider the expected value 
\begin{equation}
   \E{\#\hat{P}_{k-4}^2} = \E{\#\hat{P}_{k-4}}^2 + \var{\#\hat{P}_{k-4}} = \# P_{k-4}^2 + \var{\#\hat{P}_{k-4}}. \label{eqn2}
\end{equation}
By Lemma~\ref{lem:paths} and the fact that \( k - 4 \) is an odd number, we have \(\#P^2_{k-4} = \bigo{\delta^{k-3} n^{k-3}} \).
The variance can be decomposed into the sum of the variances of each path, which is bounded by \(\bigo{n^{k-3} / \varepsilon_1^{2k-8}}\), and the sum of covariances.
    
The covariance is non-zero only if at least two edges are shared between the two paths and all edges that appear only once exist in the original graph. As previously discussed, this forms a cycle structure, except for the path extremities that do not need to be connected. Recall the definitions of the sets \(A\) and \(B\) from the previous paragraph.
    
The set \(A\) consists of two paths at the extremities and multiple disjoint cycles. Suppose the number of edges in \(A\) is \(2p\), the number of edges in the two paths are \(q_1\) and \(q_2\), and the number of disjoint cycles is \(c\), with the number of edges in these cycles being \(r_1, \dots, r_c\). This gives us \(2p = q_1 + q_2 + \sum_{i=1}^c r_i\). In other words, \((q_1, q_2, r_1, \dots, r_c)\) forms a partition of \(2p \leq 2k\). The number of such partitions is bounded by a function of $k$. Let the bound be $f(k)$.

We now discuss the number of possible configurations of \(A\) for the partition \((q_1, q_2, r_1, \dots, r_c)\). From Lemma \ref{lem:paths} and Corollary \ref{col:numCycles}, the number of cycles of length \(q\) is bounded by \(\bigo{\delta^{q/2} n^{q/2}}\), and the number of paths of length \(q\) is bounded by \(\bigo{\delta^{q/2} n^{q/2 + 1}}\). Thus, the number of configurations for the partition \((q_1, q_2, r_1, \dots, r_c)\) is:
\[
\bigo{\delta^{q_1/2} n^{q_1/2 + 1} \cdot \delta^{q_2/2} n^{q_2/2 + 1} \cdot \prod_{t = 1}^c \delta^{r_t/2} n^{r_t/2}} = \bigo{\delta^p n^{p + 2}}.
\]
    Hence, the number of possible configurations for $A$ with $2p$ edges is no more than $f(k) \cdot \bigo{\delta^p n^{p + 2}} = \bigo{\delta^p n^{p + 2}}$. 
    
The number of edges in \(B\) is \((2k - 8 - 2p) / 2 = k - p - 4\). Using the previous argument when calculating the number of possible set $B'$, we obtain that the number of configurations for \(B\) is \(\bigo{n^{k - p - 5}}\). The number of configurations with \(|A| = 2p\) is then \(\bigo{\delta^p n^{p + 2} \cdot n^{k - p - 5}} = \bigo{\delta^p n^{k - 3}}\). Hence, the overall number of combinations is \(\sum_{p = 1}^{k - 5} \bigo{\delta^p n^{k - 3}} = \bigo{\delta^{k-5} n^{k - 3}}\).

From the previous paragraph, we observe that the covariance term outweighs the sum of the variances, leading to \( \var{\#\hat{P}_{k-4}} = \bigo{\delta^{k-5} n^{k - 3}} \). Additionally, when calculating \( \E{\#\hat{P}^2_{k - 4}} \) in~(\ref{eqn2}), it is evident that \( \#P^2_{k - 4} \) dominates \( \var{\#\hat{P}_{k-4}} \), resulting in \( \E{\#\hat{P}^2_{k - 4}} = \bigo{\delta^{k - 3}n^{k - 3}} \). Substituting \( \E{\#\hat{P}^2_{k - 4}} \) with \( \bigo{\delta^{k - 3}n^{k - 3}} \) in~(\ref{eqn1}), we find that the variance from the Laplacian mechanism is bounded by
\[
    \bigo{
        \frac{\delta^{k-2} d_{max} n^{k-2}}{\varepsilon_2^2 \varepsilon_1^4} + 
        \frac{\delta^{k-3} n^{k-2} \ln^2(n / \zeta)}{\varepsilon_2^2 \varepsilon_1^4 \varepsilon_0^2}
    }.
\]

We obtain the theorem result by summing the variance from the unbiased randomized response query and the variance from the Laplacian query with restricted sensitivity.
\end{proof}

\end{document}